%% file: main.tex
\documentclass{article}
\usepackage{arxiv} 
\usepackage[backend=biber,style=apa,sortcites=true,sorting=nyt]{biblatex}
\usepackage{csquotes}
\usepackage[utf8]{inputenc}
\usepackage[T1]{fontenc}
\usepackage{graphicx}
\graphicspath{{Fig/}{./}}   
\usepackage{subcaption}
\usepackage{multicol,multirow}
\usepackage{rotating}
\usepackage{amsmath,amssymb,amsfonts}
\usepackage{mathrsfs}
\usepackage{amsthm}
\usepackage{rotating}
\usepackage{appendix}
\usepackage{newtxtext}
\usepackage{newtxmath}
\usepackage{textcomp}
\usepackage{xcolor}
\usepackage{booktabs}
\usepackage{mathtools}
\usepackage{bbm}
\usepackage{float}
\usepackage{url}
\usepackage[colorlinks,allcolors=blue]{hyperref}

\makeatletter
\renewcommand{\@toptitlebar}{\vskip 0.25in \vskip -\parskip}
\renewcommand{\@bottomtitlebar}{\vskip 0.29in \vskip -\parskip \vskip 0.09in}
\makeatother

\newtheorem{theorem}{Theorem}[section]
\newtheorem{lemma}[theorem]{Lemma}
\newtheorem{corollary}[theorem]{Corollary}
\newtheorem{proposition}[theorem]{Proposition}
\theoremstyle{definition}
\newtheorem{remark}[theorem]{Remark}

\numberwithin{equation}{section}

\newcommand{\papertitle}{Spatial Dependencies in Item Response Theory: Gaussian Process Priors for Ecological and Cognitive Measurement}

\renewcommand{\undertitle}{Preprint}
\renewcommand{\shorttitle}{Spatial Dependencies in Item Response Theory}

\title{\papertitle}

\author{
  Mingya Huang\thanks{Corresponding author. Department of Comparative Human Development, University of Chicago. ORCID: \href{https://orcid.org/0000-0002-0647-7390}{0000-0002-0647-7390}.} \\
  University of Chicago \\
  Chicago, Illinois, USA \\
  \texttt{mhuang233@uchicago.edu} \\
\And Soham Ghosh \\
  University of Michigan \\
  Ann Arbor, Michigan, USA \\
  \texttt{ghoshsph@umich.edu} \\
}

\date{}

\begin{document}

\maketitle

\begin{abstract}
    Item Response Theory commonly assumes that responses are conditionally independent given ability and the item parameters. Even so, item difficulties may vary systematically across a space, so that nearby items have similar difficulty. That variation can carry theoretical meaning rather than measurement error alone. Citizen science platforms, where difficulty varies with geographic region, and Author Recognition Tests, where knowledge is organized by genre, both show this. Existing spatial IRT methods model such structure through a conditional autoregressive prior on the item parameters \parencite{Can_ado_2016, santos2020bayesian}. They focus on binary responses, fix the covariance structure in advance, and assume that item coordinates are observed. We propose Spatial Gaussian Process Item Response Theory (SGP-IRT), which retains conditional independence of the responses while placing a Gaussian process prior on item difficulty. The prior learns that structure from the data and provides joint uncertainty for the model parameters. It also accommodates polytomous responses and locates items in geographic coordinates or a latent similarity space recovered from responses. In simulations, it improves item difficulty recovery relative to existing spatial IRT methods. Applied to Author Recognition Test data, it recovers a difficulty surface that can be read for substantive meaning.
\end{abstract}

\keywords{Item response theory \and spatial dependency \and Gaussian processes \and latent variable modeling \and Bayesian nonparametrics}

\section{Introduction}

Item Response Theory (IRT) provides a foundation for measurement across multiple domains including psychology, education, and ecology. It typically assumes that responses are conditionally independent once individual ability and the item parameters are accounted for \parencite{Stout2002,Fox2010}. Even when this local independence assumption holds, however, item difficulties may vary systematically across a geographic, cognitive, or conceptual space. Standard IRT models estimate those difficulties separately and therefore do not exploit this structured heterogeneity \parencite{Kang2024,Brucato2023}. Recent work in latent space modeling addresses a distinct problem by embedding respondents and items in a shared metric space to capture interactions that standard IRT models miss \parencite{Jeon_2021, Jin_2019}. Our focus is instead on settings where the item parameters themselves exhibit meaningful spatial structure while responses remain conditionally independent given ability and the item parameters.

\subsection{Motivating Examples: Ecological \& Cognitive Assessments}

In human-ecological assessment, participants on citizen science platforms classify field data that researchers cannot label at scale on their own. A representative case is Snapshot Serengeti, where camera traps across a wildlife reserve photograph passing animals and online volunteers identify the species shown. Each volunteer acts as a respondent, each camera location as an item, and a correct species identification as the binary response. An item's difficulty then reflects the cognitive demands of species identification at that location. This difficulty varies systematically across the reserve, with environmental conditions, image characteristics, and which species are common at each site. Figure \ref{fig:ecological_motivation} illustrates this with Snapshot Serengeti camera trap data. The raw difficulties vary across the reserve, and a smoothed version of them suggests a coherent east-west gradient rather than variation without pattern. A gradient of this kind is consistent with environmental and perceptual factors, which suggests that spatial variation in item difficulty should not automatically be treated as unstructured item-to-item noise. Here, it corresponds to real differences in cognitive demands across geographic contexts, information that can improve both ability estimation and understanding of the task structure.

\begin{figure}[H]
    \centering
    \includegraphics[width=\linewidth]{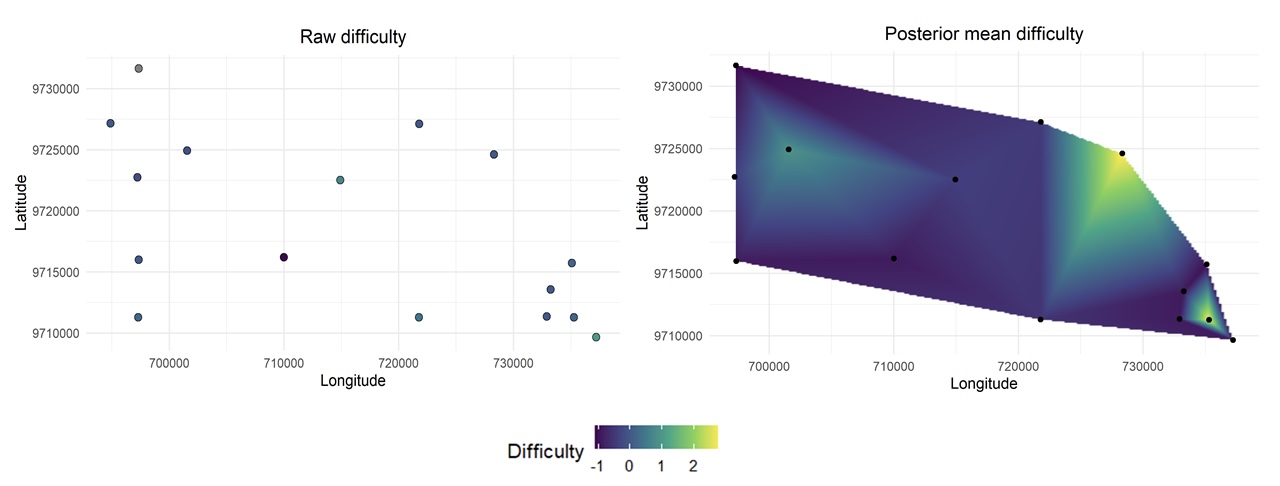}
    \caption{Spatial patterns in classification difficulty from human-ecological assessment data. Left panel: Raw difficulty computed as the proportion of incorrect species identifications by human raters at each camera location, showing considerable spatial variation across the study region. Right panel: A spatially smoothed version of the same difficulties, obtained with the model developed in Section \ref{sec:method} and shown here to motivate the problem rather than as a result. The smoothed surface suggests a coherent east-west gradient, a pattern consistent with environmental and perceptual influences on classification difficulty.}
    \label{fig:ecological_motivation}
\end{figure}

A parallel pattern arises in cognitive assessment, where items may group by psychological construct rather than physical location. Author Recognition Tests (ARTs) are widely used to measure print exposure and reading experience \parencite{StratovichART,Moore2015,Wimmer2023,acheson2008}. Responses to these tests often show organization by literary genre. Individuals who recognize science fiction authors such as Ursula K. Le Guin tend to have higher recognition rates for others in the same genre, such as Isaac Asimov. This grouping is not incidental. Recognition draws on memory retrieval, knowledge organized through literary schemas, and metacognitive judgments about author familiarity, and those processes operate within genres. Genre membership therefore supplies a similarity structure over which item difficulty may vary. ARTs are in this sense not merely a checklist of recognized names. They also carry information about how domain knowledge is structured in memory, with direct implications for test construction, scoring, and interpretation. Figure~\ref{fig:art_motivation} previews that structure for one such test. It places each item by response similarity and shades it by estimated difficulty, using the model we develop in Section~\ref{sec:method}. Difficulty then varies smoothly within each genre rather than item by item, the same kind of structured variation the ecological example shows.

\begin{figure}[h]
    \centering
    \includegraphics[width=\linewidth]{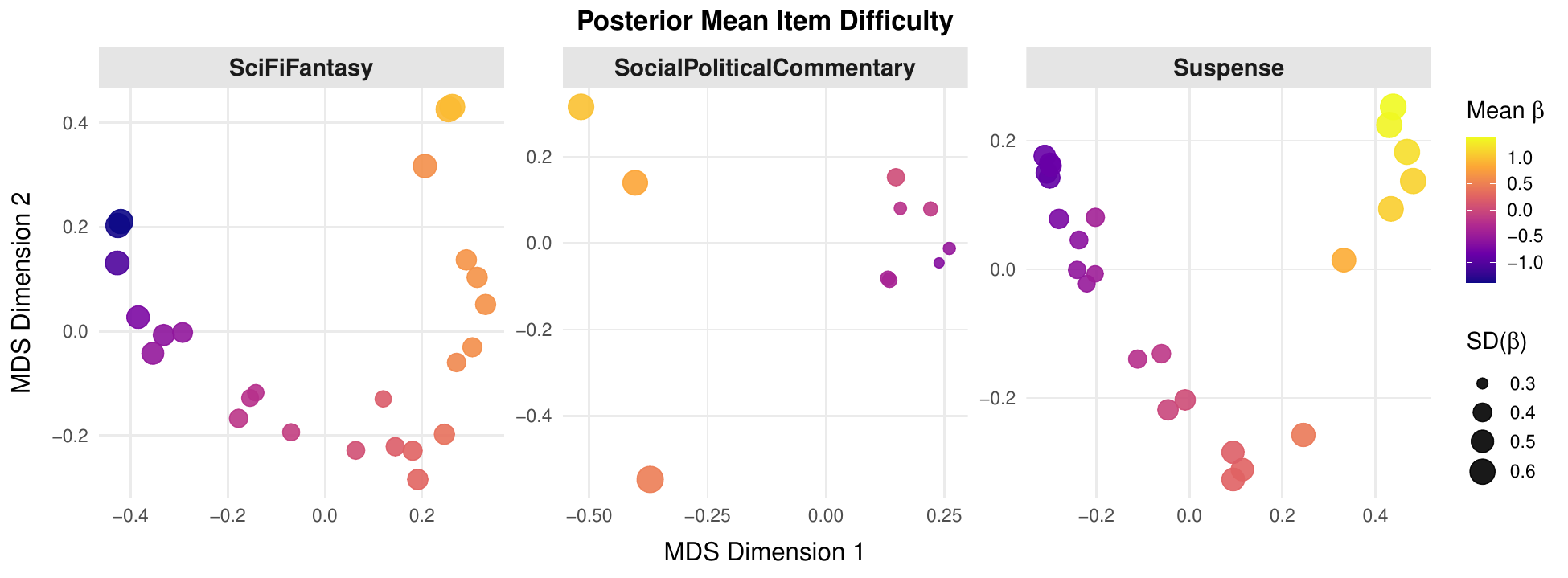}
    \caption{Estimated item difficulty for an Author Recognition Test, shown on a two-dimensional multidimensional scaling (MDS) embedding of item response similarity, faceted by literary genre. Items answered in similar ways sit close together. Color denotes the estimated difficulty (the posterior mean of the intercept-scale parameter \(\beta_j\), purple easier and yellow harder) and point size denotes the uncertainty in that estimate (its posterior standard deviation), both taken from the model developed in Section~\ref{sec:method}. Within each genre, the map provides a descriptive view of how difficulty varies across the estimated item geometry.}
    \label{fig:art_motivation}
\end{figure}

\subsection{Current Challenges}
Both examples raise a central question for measurement. When item difficulties vary over a space, a model must decide whether that variation is unstructured item-to-item noise or a systematic feature worth recovering. This question is distinct from whether responses remain conditionally independent after ability and the item parameters are specified. The difficulty is to model structured item heterogeneity without giving up the precision and interpretability that make IRT useful. Existing methods differ partly because they target these two distinct sources of structure.

One group treats local independence violations as nuisance structure to be removed, with the goal of protecting the ability estimate rather than interpreting the dependence \parencite[e.g.,][]{Wainer_1987, Cai_2011, Hoskens_1997}. Testlet models, for instance, add a person-specific effect for a group of items that share a source of dependence, such as several questions tied to one reading passage \parencite{Wainer_1987, Bradlow_1999}. The shared effect keeps that dependence from distorting the ability estimate. Bifactor models instead give each item two loadings, one on a general trait and one on a narrower group factor that captures association within its cluster \parencite{Gibbons_1992, Cai_2011}. Residual correlation approaches keep the base model but let selected item pairs correlate through their remaining variance \parencite{Yen_1984, Hoskens_1997}. The three differ in implementation, yet share one goal. Each separates the dependence from ability, restoring conditional independence.

A second group addresses a different problem: it retains conditional independence of the responses while building spatial structure directly into the item parameters. The spatial IRT family \parencite{Can_ado_2016, santos2020bayesian} places a conditional autoregressive (CAR) prior on item difficulty, with later variants adding discrimination and guessing parameters. The prior lets items at nearby locations share information. Unlike the first group, it regularizes structured heterogeneity in the item parameters rather than modeling residual response dependence, which is closer to our own view.

Both groups, however, leave needs unmet for the practical settings we consider. The testlet and bifactor models require the dependent items to be grouped in advance, as a partition into testlets or content clusters. The residual-correlation approach instead estimates each pair on its own, with no structure linking the pairs. These methods are built to address residual response dependence. That suits dependence from test design, such as shared wording or a common passage, but it does not recover spatially structured heterogeneity in item difficulty. The spatial IRT family does model such heterogeneity, yet its CAR prior confines association to neighboring locations and holds the pattern the same in every direction. None of these methods was designed for items on a continuous surface, where the structure worth recovering is a smooth, possibly informative gradient in difficulty.

Three gaps follow from this. The current spatial IRT methods are restricted to binary responses, which sets aside the rating scales and partial-credit formats common in current assessments. Their item-parameter covariance structure is also fixed in advance rather than learned from the data. Information is then not shared across the full surface, since the correction methods estimate residual dependence pair by pair or within fixed groups, and even the spatial prior reaches only as far as its neighbor graph.

\subsection{Our Contributions}
To our knowledge, these three needs have not been addressed together in a single model. We take a step toward this with Spatial Gaussian Process Item Response Theory (\texttt{SGP-IRT}), which places a Gaussian process prior on item parameters. This retains conditional independence of the responses while treating spatially structured variation in item difficulty as part of the measurement model. The first three contributions are methodological and answer the three challenges in turn; the fourth turns to applications.

First, \texttt{SGP-IRT} accommodates polytomous responses through a fixed-score reference-category softmax likelihood, so the same spatial prior serves binary and multi-category tests. Corollary \ref{cor:poly_contraction} shows posterior contraction for any fixed number of response categories, which places the extension on the same theoretical footing as the binary case (i.e., Theorem \ref{thm:sgp_irt_contraction}). We therefore develop the polytomous case at the level of the model and its theory, and the studies reported below stay binary, which matches the response format the existing spatial IRT methods were built for.

Second, \texttt{SGP-IRT} learns the covariance structure of the difficulty surface from the data rather than fixing it in advance. Its anisotropic squared-exponential prior estimates a separate length-scale for each dimension, so the data determine how far association among item difficulties extends and whether it differs by direction. \texttt{SGP-IRT} attains the lowest difficulty RMSE in every simulation scenario in Section \ref{sec:simulation}, including one generated from a CAR process, and its interval coverage stays near the nominal level where the competitors undercover. The theory points the same way. Under smoothness assumptions, Theorem \ref{thm:sgp_irt_contraction} shows that the GP posterior for the difficulty surface concentrates on the truth as respondents and items grow together, even while ability and discrimination are estimated jointly.

Third, \texttt{SGP-IRT} shares information across locations, with all parameters estimated in one posterior. A single global covariance links all items, so each difficulty estimate borrows strength from the full surface and the posterior reports uncertainty for the parameters jointly. Earlier methods approach the relevant structures differently, modeling residual response association pair by pair or within fixed groups, or sharing item-parameter information across a neighbor graph. For the application demonstrated in Section \ref{sec:empirical},  we fit each genre in ARTs under its own covariance, so the sharing operates within a genre rather than across the whole item pool.

Fourth, \texttt{SGP-IRT} broadens where spatial IRT can be used in practice. When items carry no geographic coordinates, it recovers a latent similarity space in which items close in meaning sit close together. The Author Recognition Test in Section~\ref{sec:empirical} shows this, and the same prior applies where coordinates are observed, as in neighborhood measurement \parencite{Raudenbush_1999, VerbitskySavitz2009}. In either setting, the surface carries substantive meaning rather than serving only as calibration, since it shows where a task is harder than its neighbors suggest. We view the model as a starting point, one that also opens screening for differential item functioning without predefined groups.

The remainder of the paper is structured as follows. We review IRT fundamentals and existing spatial approaches in Section \ref{sec:background}, then introduce \texttt{SGP-IRT} in Section \ref{sec:method}, and present posterior contraction results in Section \ref{sec:theory}. Next, Section \ref{sec:simulation} reports simulation studies whose design follows the ecological setting, where item coordinates are observed. We then apply \texttt{SGP-IRT} to Author Recognition Test data, where the coordinates are recovered from responses in Section \ref{sec:empirical}, and close with implications for test development and adaptive testing in Section \ref{sec:disc}.

\section{Background}\label{sec:background}

\subsection{A Brief Review of Item Response Theory}
Item Response Theory (IRT) provides a probabilistic model for the relationship between respondents' latent traits and their responses to test items \parencite{lord1968statistical, hambleton1985item}. More specifically, IRT specifies that the probability of a particular response to an item depends on both the respondent's ability level and the item's characteristics \parencite{yang2014IRT}. This relationship is formalized through item response functions that map latent traits to response probabilities. The classic IRT model for dichotomous responses takes the form:
\begin{equation}
\mathbb{P}(Y_{ij} = 1 | \theta_i, \boldsymbol{\xi}_j) = f(\theta_i; \boldsymbol{\xi}_j),
\end{equation}
where $Y_{ij}$ represents the response (0 or 1) of respondent $i$ to item $j$, $\theta_i$ denotes the respondent's ability, and $\boldsymbol{\xi}_j$ contains the item parameters.

The Rasch model (1PL) represents the simplest IRT formulation \parencite{Rasch1961}:
\begin{equation}
\mathbb{P}(Y_{ij} = 1 | \theta_i, b_j) = \frac{\exp(\theta_i - b_j)}{1 + \exp(\theta_i - b_j)},
\end{equation}
where $b_j$ is the item difficulty parameter. This parameter indicates the ability level at which respondents have a 0.5 probability of answering correctly. The 1PL model assumes all items have equal discrimination and no guessing \parencite{hambleton1985item}.

The two-parameter logistic (2PL) model, developed by \textcite{birnbaum1968}, extends this model:
\begin{equation}
\mathbb{P}(Y_{ij} = 1 | \theta_i, \alpha_j, b_j) = \frac{\exp[\alpha_j(\theta_i - b_j)]}{1 + \exp[\alpha_j(\theta_i - b_j)]},
\end{equation}
where $\alpha_j$ represents the item discrimination parameter. This parameter determines how sharply the probability of a correct response changes with ability near the item's difficulty point \parencite{lord1968statistical}.

For items where guessing may occur (i.e., multiple-choice questions), the three-parameter logistic (3PL) model adds a lower asymptote \parencite{birnbaum1968}:
\begin{equation}
\mathbb{P}(Y_{ij} = 1 | \theta_i, \alpha_j, b_j, c_j) = c_j + (1-c_j)\frac{\exp[\alpha_j(\theta_i - b_j)]}{1 + \exp[\alpha_j(\theta_i - b_j)]},
\end{equation}
where $c_j$ represents the pseudo-guessing parameter. This parameter reflects the probability of a correct response for respondents with extremely low ability. The models developed below carry no guessing component, so $c_j$ does not appear again.

These models share the local independence assumption: conditional on ability and the item parameters, item responses are taken to be independent. Violations of this assumption are commonly addressed through testlet effects, group factors, or residual associations. A distinct extension retains local independence while allowing the item parameters to vary systematically across a space, so that nearby items can share information.

\subsection{Spatial Extensions of IRT}

A distinct spatial extension retains local independence of the responses while allowing item parameters to depend on location. The spatial IRT family proposed by \textcite{Can_ado_2016} and \textcite{santos2020bayesian} does this through a Bayesian prior on the item parameters. In the \texttt{1PLUS} model, item difficulty varies across sites; \texttt{2PLUS} adds spatially varying discrimination, and \texttt{3PLUS} adds a guessing parameter. Writing $\mathbf b = (b_1,\dots,b_J)^\top$ for the vector of item difficulties across the $J$ sites, these models tie the entries together through a conditional autoregressive (CAR) prior,
\begin{equation}\label{eq:car-prior}
\mathbf b \mid \tau \sim \mathcal{N}_J\bigl(\mathbf 0,\ [\tau(\mathbf D - \mathbf W)]^{-1}\bigr),
\end{equation}
where $\mathbf W=(w_{jj'})$ is an adjacency matrix with $w_{jj'}=1$ when sites $j$ and $j'$ are neighbors and $0$ otherwise, $\mathbf D=\mathrm{diag}(w_{1+},\dots,w_{J+})$ collects the neighbor counts $w_{j+}=\sum_{j'} w_{jj'}$, and the precision $\tau>0$ scales the variance. This is the intrinsic CAR (an intrinsic Gaussian Markov random field), in which each difficulty is centered on the average of its neighbors, $b_j \mid b_{-j}, \tau \sim \mathcal N\!\bigl(w_{j+}^{-1}\sum_{j'\sim j} b_{j'},\, (\tau w_{j+})^{-1}\bigr)$. Adjacent sites are pulled toward one another, and prior association enters only through the adjacency encoded in $\mathbf W$. The same construction has been adapted outside measurement, for instance to spatial cluster detection \parencite{Can_ado_2016}, where it again links neighboring locations through an adjacency graph.

A frequentist alternative, geographically weighted IRT, lets item parameters vary across respondent locations by fitting a separate model in local windows \parencite{Fotheringham2002, Robinson2017}. Its space is indexed by respondents rather than by items, so it answers a different question and needs respondent coordinates that our applications do not supply. The CAR family stays closer to our setting, since it places the spatial structure in the item difficulties, the component our method reconsiders. Its adjacency graph fixes the strength of that prior association in advance, which leaves open how far association should reach when items lie on a continuous surface.

\section{Our Proposed Method: \texttt{SGP-IRT}}
\label{sec:method}

Spatial Gaussian Process Item Response Theory (\texttt{SGP-IRT}) answers that question by pairing an IRT observation model with a Gaussian Process prior on the item-difficulty surface. The prior expects, rather than requires, that items close together in an observed or prespecified item space have similar difficulties. That expectation is what allows information to be shared across nearby items. Two features determine its flexibility: (i) the length scales are estimated, so the data decide how far similarity extends in each direction; and (ii) the nugget, a term from spatial statistics for variation that neighboring locations do not share, absorbs item-specific departures from the smooth surface. Both are learned from the data, so difficulty need not be a smooth function of position alone. A true surface may instead vary sharply between neighboring items, or shift abruptly across a boundary that divides the space. In such cases, the fitted model attributes less variation to the smooth component and more to the nugget. Section~\ref{sec:simulation} reports how well that adjustment recovers the underlying difficulties.


\paragraph{Model specification.}
\label{subsec:observation_model}

Let \(I\) be the number of respondents, \(J\) the number of items, and \(\mathcal O\subseteq\{1,\ldots,I\}\times\{1,\ldots,J\}\) the set of observed person--item pairs.  We write \(N=|\mathcal O|\), so the notation also accommodates incomplete response matrices.  The response \(Y_{ij}\), the ability \(\theta_i\), and the discrimination \(\alpha_j>0\) keep the meanings given in Section~\ref{sec:background}.  Item difficulty is carried by an intercept-scale parameter \(\beta_j\), introduced with the observation model below.  Positivity of the discrimination is imposed on the log scale, as Eq.~\eqref{eq:discrimination_prior} specifies.  We collect the item parameters in \(\boldsymbol\alpha=(\alpha_1,\ldots,\alpha_J)^\top\) and \(\boldsymbol\beta=(\beta_1,\ldots,\beta_J)^\top\).  Each item also carries a coordinate \(\mathbf s_j=(s_{j1},\ldots,s_{jD})^\top\in\mathcal S\subset\mathbb R^D\), which the classical models do not use.  Observation-level covariates enter through \(\mathbf x_{ij}\in\mathbb R^p\), a centered vector with coefficients \(\boldsymbol\gamma\).  The covariate vector carries no intercept. The local independence assumed in Section~\ref{sec:background} is retained in full,
\[
  Y_{ij}\ \perp\!\!\perp\ Y_{ik}
  \mid
  \theta_i,\boldsymbol\alpha,\boldsymbol\beta,
  \boldsymbol\gamma,
  \quad j\neq k .
\]
Spatial structure therefore enters through the prior on \(\boldsymbol\beta\) rather than through the response distribution.

The primary interpretation of \(\mathbf s_j\) is as an observed or independently defined item location.  When geographic coordinates are not available, latent coordinates can instead be estimated from the response data.  The natural starting point is to place two items close together when respondents answer them alike.  Correlating the raw responses would not capture that.  Two easy items are answered alike by almost everyone, and strong respondents answer nearly every item correctly, so raw agreement mostly reflects difficulty and ability.  We therefore remove the respondent and item main effects before measuring association.  A nonspatial 2PL model, the same observation model with no spatial structure on the difficulties, gives the posterior expected probabilities \( \widehat p_{ij} =\mathbb E\{\operatorname{logit}^{-1}(\alpha_j\theta_i-\beta_j +\mathbf x_{ij}^\top\boldsymbol\gamma)\mid\mathbf Y\}. \)  Standardized residuals \( r_{ij}=(Y_{ij}-\widehat p_{ij}) /\{\widehat p_{ij}(1-\widehat p_{ij})\}^{1/2} \) then carry what those effects leave behind.  Writing \(\widehat C_{jk}=\operatorname{cor}(r_{\cdot j},r_{\cdot k})\) for the residual item-association matrix, we apply classical multidimensional scaling (MDS; \textcite{BujaMDS}) to the dissimilarities \(d_{jk}=\{2(1-\widehat C_{jk})\}^{1/2}\).  Two items are therefore close when respondents depart from the nonspatial fit in similar ways.  In predictive comparisons, \(\widehat C_{jk}\) and the MDS configuration are constructed using the training respondents only and are then held fixed for the held-out respondents.  The resulting posterior is conditional on the estimated coordinates and does not propagate embedding uncertainty.  Likewise, the theory in Section \ref{sec:theory} treats \(\mathbf s_1,\ldots,\mathbf s_J\) as fixed and known.

Difficulties at nearby coordinates are tied together by a Gaussian process, a collection of random variables \(\{f(\mathbf s):\mathbf s\in\mathcal S\}\) such that every finite vector \(\mathbf f=(f(\mathbf s_1),\ldots,f(\mathbf s_m))^\top\) is multivariate normal, \(\mathbf f\sim\mathcal N_m(\boldsymbol\mu,\boldsymbol\Sigma)\) \parencite{rasmussen03}.  The mean function \(\mu(\mathbf s)=\mathbb E\{f(\mathbf s)\}\) is often set to zero, while the positive-definite kernel \(k(\mathbf s,\mathbf s')=\operatorname{cov}\{f(\mathbf s),f(\mathbf s')\}\) encodes assumptions about features such as smoothness and correlation range. A Gaussian process thus supplies a nonparametric prior for a latent function and yields joint posterior uncertainty conditional on the model and the item coordinates \parencite{neal97}.

For binary responses \(Y_{ij}\in\{0,1\}\), our observation model is
\begin{align}
  Y_{ij}\mid\theta_i,\alpha_j,\beta_j,\boldsymbol\gamma
  &\sim
  \operatorname{Bernoulli}(\pi_{ij}),
  \quad (i,j)\in\mathcal O,
  \label{eq:sgp_binary_obs}\\
  \operatorname{logit}(\pi_{ij})
  &=
  \eta_{ij}
  =
  \alpha_j\theta_i-\beta_j
  +\mathbf x_{ij}^{\top}\boldsymbol\gamma .
  \label{eq:sgp_binary_predictor}
\end{align}
Difficulty enters with a negative sign, so larger \(\beta_j\) means a lower probability of a correct response at a given ability.  We call \(\beta_j\) an intercept-scale difficulty because it enters the predictor additively rather than as a point on the ability scale.  Expanding the 2PL of Section~\ref{sec:background} gives \(\alpha_j(\theta_i-b_j)=\alpha_j\theta_i-\beta_j\), so that \(\beta_j=\alpha_jb_j\) while the discrimination is unchanged.  The two scales rank items differently: \(\beta_j\) at average ability, and \(b_j\) at the point where the success probability is \(0.5\).  We use this same linear predictor for fitting, posterior prediction, and every matched model comparison.

Let \(C\geq2\) denote the number of response categories, a fixed constant, and let \(c\) index a single category.  With \(Y_{ij}\in\{0,\ldots,C-1\}\), we take category \(0\) as the reference and set \(\eta_{ij0}=0\).  For \(c=1,\ldots,C-1\), we define
\begin{equation}
  \eta_{ijc}
  =
  q_c\alpha_j\theta_i-\beta_{jc}
  +\mathbf x_{ij}^{\top}\boldsymbol\gamma_c ,
  \label{eq:sgp_poly_predictor}
\end{equation}
where \(q_1,\ldots,q_{C-1}\) are fixed, prespecified, nonzero category scores.  Each \(q_c\) sets how strongly ability separates category \(c\) from the reference, so the scores fix the spacing of the categories along the ability scale rather than estimating it.  Taking \(q_c=c\) gives equal spacing.  The reference convention \(\eta_{ij0}=0\) is equivalent to \(q_0=0\), \(\beta_{j0}=0\), and \(\boldsymbol\gamma_0=\mathbf0\).  The response probabilities are
\begin{align}
  \mathbb{P}(Y_{ij}=c\mid\cdot)
  &=
  \frac{\exp(\eta_{ijc})}
       {1+\sum_{k=1}^{C-1}\exp(\eta_{ijk})},
   \ c=1,\ldots,C-1,
  \label{eq:sgp_poly_prob}\\
  \mathbb{P}(Y_{ij}=0\mid\cdot)
  &= \frac{1}
       {1+\sum_{k=1}^{C-1}\exp(\eta_{ijk})}.
  \label{eq:sgp_poly_reference}
\end{align}

Category-specific coefficients \(\boldsymbol\gamma_c\) are required because adding the same covariate term to every softmax logit would cancel from the response probabilities.  When \(C=2\) and \(q_1=1\), Eqs. \eqref{eq:sgp_poly_predictor}--\eqref{eq:sgp_poly_reference} reduce to the binary model in Eqs. \eqref{eq:sgp_binary_obs}--\eqref{eq:sgp_binary_predictor}.  This is an identified fixed-score reference-category model.  The parameters \(\beta_{jc}\) are category-specific intercept-scale difficulties; in the absence of ordering constraints, they should not be interpreted as ordered thresholds.  Accordingly, we do not claim that Eqs. \eqref{eq:sgp_poly_predictor}--\eqref{eq:sgp_poly_reference} simultaneously define graded-response, partial-credit, and unrestricted nominal-response models, each of which requires its own likelihood and identification constraints.

\paragraph{Covariates and parameter interpretation.}
Covariates are not part of the classical 1PL or 2PL formulation, but they are useful in explanatory IRT models for separating systematic response variation from latent ability and item difficulty. The linear logistic test model attaches covariates to the item difficulties, which confines them to item-level characteristics. Placing them in the linear predictor instead keeps the respondent-by-item component available. Our notation \(\mathbf x_{ij}\) is deliberately general and does not require every covariate to vary across both respondents and items. In particular, it may be decomposed as \( \mathbf x_{ij} = \left( \mathbf x_i^{\mathrm P\top}, \mathbf x_j^{\mathrm I\top}, \mathbf x_{ij}^{\mathrm{PI}\top} \right)^\top, \) where \(\mathbf x_i^{\mathrm P}\) contains respondent-level characteristics, \(\mathbf x_j^{\mathrm I}\) contains item-level characteristics, and \(\mathbf x_{ij}^{\mathrm{PI}}\) contains genuinely respondent-by-item or administration-level characteristics. Respondent-level examples include age, education, vocabulary score, language background, clinical group, or prior domain experience. Item-level examples include text length, word frequency, item format, content domain, image quality, author popularity, or the position of an item in a test. Respondent-by-item examples include whether a respondent has previously encountered the material represented by an item, respondent--item language or domain matching, item presentation order, testing condition, or time pressure.

\paragraph{Gaussian-process prior on item difficulty.}
\label{subsec:spatial_prior}
For each nonreference category \(c\), we decompose the intercept-scale difficulty as
\begin{equation}
  \beta_{jc}
  =
  f_c(\mathbf s_j)+u_{jc},
  \quad
  u_{jc}\stackrel{\mathrm{iid}}{\sim}
  \mathcal N(0,\sigma_{u,c}^2),
  \label{eq:difficulty_decomposition}
\end{equation}
where \(f_c\) is a smooth spatial component and \(u_{jc}\) is residual item-specific variation.  In the binary model we suppress the category subscript and write \(\beta_j=f(\mathbf s_j)+u_j\).  Independently across the nonreference categories, we assign
\begin{align}
  f_c\mid\sigma_{f,c},\boldsymbol\ell_c
  &\sim
  \operatorname{GP}\!\left(
    0,\,
    \sigma_{f,c}^2K_{\boldsymbol\ell_c}
  \right),
  \label{eq:gp_prior}\\
  K_{\boldsymbol\ell_c}(\mathbf s,\mathbf t)
  &=
  \exp\!\left\{
    -\frac12\sum_{d=1}^{D}
    \frac{(s_d-t_d)^2}{\ell_{cd}^2}
  \right\},
  \label{eq:ard_kernel}
\end{align}
where \(\boldsymbol\ell_c=(\ell_{c1},\ldots,\ell_{cD})^\top\) contains direction-specific length scales and \(\sigma_{f,c}>0\) is the marginal standard deviation of the smooth surface.  Equivalently, after integrating out the two latent components,
\begin{equation}
  \boldsymbol\beta_{\cdot c}
  \mid
  \sigma_{f,c},\boldsymbol\ell_c,\sigma_{u,c}
  \sim
  \mathcal N_J\!\left(
    \mathbf0,\,
    \sigma_{f,c}^2\mathbf K_{\boldsymbol\ell_c}
    +\sigma_{u,c}^2\mathbf I_J
  \right).
  \label{eq:marginal_difficulty_prior}
\end{equation}

Allowing a separate length scale in each direction gives the automatic relevance determination (ARD) kernel of \textcite{mackay96}.  A small \(\ell_{cd}\) permits rapid variation along direction \(d\), whereas a large \(\ell_{cd}\) implies slow decay of correlation along that direction.  This flexibility can improve regularization when the surface is genuinely anisotropic, but it does not mean that a visibly smooth map by itself establishes anisotropy.  Moreover, when the coordinates are latent rather than geographic, the direction-specific length scales are conditional on the orientation of the selected embedding and are interpreted primarily as regularization parameters.

The nugget standard deviation \(\sigma_{u,c}\) distinguishes smooth spatial variation from item-specific variation.  Define the nugget-to-total-variance ratio
\begin{equation}
  r_{u,c}
  =
  \frac{\sigma_{u,c}^2}
       {\sigma_{f,c}^2+\sigma_{u,c}^2}.
  \label{eq:sgp_nugget_ratio}
\end{equation}
For \(j\neq j'\), the prior correlation between two item difficulties is
\begin{equation}
  \operatorname{cor}(\beta_{jc},\beta_{j'c})
  =
  (1-r_{u,c})
  K_{\boldsymbol\ell_c}(\mathbf s_j,\mathbf s_{j'}).
  \label{eq:sgp_nugget_correlation}
\end{equation}
Thus, the length scales determine how correlation decays with distance, \(\sigma_{f,c}^2\) determines the smooth spatial variation, and the nugget attenuates correlation at every nonzero separation.  When \(r_{u,c}\) is near zero, the model strongly favors a smooth surface; when it is near one, the model approaches an item-specific random-effects specification and spatial borrowing is weak.

Before model fitting, each coordinate direction is rescaled to \([0,1]\). Writing \(A_{cd}=\ell_{cd}^{-1}\), we use the hyperpriors
\begin{equation}
  A_{cd}^{D} \sim\operatorname{Gamma}(a_A,b_A),
  \quad
  \sigma_{f,c},\sigma_{u,c}
  \sim\operatorname{HalfNormal}(0,1),
  \label{eq:gp_hyperpriors}
\end{equation}
independently across \(c\) and \(d\), where the Gamma distribution is parameterized by shape and rate and \(a_A,b_A>0\) are fixed.  The inverse length-scale prior matches the rescaled squared-exponential GP used in the theoretical analysis (Assumption (A5), Section \ref{sec:theory}).

\paragraph{Identification.}
\label{subsec:identification}
The IRT likelihood has location, scale, and sign indeterminacies.  We remove them by imposing
\begin{equation}
  \frac1I\sum_{i=1}^{I}\theta_i=0,
  \quad
  \frac1I\sum_{i=1}^{I}\theta_i^2=1,
  \quad
  \alpha_j>0,\quad j=1,\ldots,J.
  \label{eq:identification_constraints}
\end{equation}
For computation, we draw \(z_i\stackrel{\mathrm{iid}}{\sim}\mathcal N(0,1)\) and define
\begin{equation}
  \theta_i
  =
  \frac{z_i-\bar z}
       {\left\{I^{-1}\sum_{r=1}^{I}(z_r-\bar z)^2\right\}^{1/2}},
  \quad
  \bar z=I^{-1}\sum_{r=1}^{I}z_r,
  \label{eq:standardized_ability}
\end{equation}
which enforces the first two constraints exactly.  Positivity of the discriminations is enforced by working on the log scale, as Eq.~\eqref{eq:discrimination_prior} sets out.  Together with the reference category of Eq.~\eqref{eq:sgp_poly_predictor} and the absence of a covariate intercept, these devices identify the stated parameterization.  We therefore do not impose an additional sum-to-zero penalty on each difficulty vector, since fixing the ability mean already pins the level of the difficulties.

\paragraph{Remaining priors.}
The log discriminations receive a hierarchical normal prior,
\begin{equation}
  a_j=\log\alpha_j,
  \quad
  a_j\mid\sigma_a
  \stackrel{\mathrm{iid}}{\sim}
  \mathcal N(0,\sigma_a^2),
  \quad
  \sigma_a\sim\operatorname{HalfNormal}(0,1).
  \label{eq:discrimination_prior}
\end{equation}
The prior median of \(\alpha_j\) is therefore one, and negative discriminations receive zero prior probability.  For a fixed weakly informative scale \(\sigma_\gamma>0\), we set
\begin{equation}
  \boldsymbol\gamma_c
  \sim
  \mathcal N_p(\mathbf0,\sigma_\gamma^2\mathbf I_p),
  \quad c=1,\ldots,C-1.
  \label{eq:covariate_prior}
\end{equation}
In the binary model, \(\boldsymbol\gamma=\boldsymbol\gamma_1\).

\paragraph{Distinction from existing spatial IRT approaches.}
\label{subsec:model_comparison}
Our specification retains the IRT observation likelihood and changes only the prior geometry used to regularize item difficulty.  The CAR prior in Eq.~\eqref{eq:car-prior} has a sparse precision matrix, and association between distant sites arises by propagation through the graph. Eqs.~\eqref{eq:gp_prior}--\eqref{eq:ard_kernel} instead express covariance as a continuous function of distance and estimate a correlation range for each direction.  That flexibility carries a computational cost, since the dense covariance factorization is \(\mathcal O(J^3)\) per nonreference category in the implementation used here.  Neither prior dominates the other, as a Gaussian process suits smooth distance-based structure while a CAR or another roughness prior may suit graph-aligned or discontinuous structure.  Our empirical comparisons therefore hold the likelihood, covariates, priors on shared parameters, posterior summaries, and computational settings fixed, and change only the prior on item difficulty.

The closest citizen-science model is \texttt{3PLUS} of \textcite{santos2020bayesian}, developed for binary camera-trap accuracy data indexed by site and species.  That model decomposes difficulty into site and species effects, with 1PL, 2PL, and 3PL variants adding discrimination and guessing parameters.  \texttt{SGP-IRT} uses a single spatially regularized difficulty surface with an item-specific discrimination and no guessing parameter, so the comparison is not like-for-like for every estimand.  We therefore do not report discrimination recovery for \texttt{1PLUS}, which fixes a common discrimination, and we do not attribute the behavior of \texttt{3PLUS} under no-guessing data solely to its CAR component.  The geographically weighted approach of \textcite{Robinson2017} pools within local windows rather than under one covariance model, as noted in Section \ref{sec:background}.
\paragraph{Estimation, diagnostics, and uncertainty quantification.}
\label{subsec:computation}
We implement \texttt{SGP-IRT} in \texttt{Stan} and sample from the posterior using the No-U-Turn Sampler (NUTS) \parencite{hoffman2011nouturnsampleradaptivelysetting}. For each nonreference category, we use the noncentered representation
\begin{equation}
\mathbf f_c
=
\sigma_{f,c}\mathbf L_c\mathbf z_{f,c},
\quad
\mathbf u_c
=
\sigma_{u,c}\mathbf z_{u,c},
\quad
\mathbf z_{f,c},\mathbf z_{u,c}
\stackrel{\mathrm{ind}}{\sim}
\mathcal N_J(\mathbf0,\mathbf I_J),
\label{eq:sgp_noncentered_gp}
\end{equation}
where \(\mathbf L_c\mathbf L_c^\top=\mathbf K_{\boldsymbol\ell_c}\) and \(\boldsymbol\beta_{\cdot c}=\mathbf f_c+\mathbf u_c\). This parameterization separates the smooth and item-specific components and avoids directly sampling a highly correlated difficulty vector.

Unless otherwise stated, we run four chains for 2,000 iterations each, discarding the first 1,000 iterations of each chain as warmup and retaining 4,000 post-warmup draws in total. We use $\texttt{adapt\_delta}=0.99,$ and $\texttt{max\_treedepth}=15.$ Convergence is evaluated using the rank-normalized split \(\widehat R\), bulk and tail effective sample sizes (ESS), divergent transitions, maximum-treedepth exceedances, energy Bayesian fraction of missing information (E-BFMI), and visual inspection of selected trace plots \parencite{Vehtari2021}. We require \(\widehat R<1.01\), flag scalar parameters with bulk or tail ESS below 400, require E-BFMI to exceed 0.30 in every chain, and exclude fits with unresolved divergences or persistent trace-plot pathologies from substantive comparisons.

The three full-data ART fits in Section \ref{sec:empirical} meet every one of these criteria. Numerical diagnostics and representative trace plots are reported in Appendix \ref{app:art_mcmc}.

\section{Theoretical Guarantees}\label{sec:theory}
We establish posterior contraction for the smooth item-difficulty surface under the Bernoulli IRT likelihood itself. The result below allows \(I\) and \(J\) to grow jointly and estimates ability, discrimination, and the difficulty surface together.

Theorem \ref{thm:sgp_irt_contraction} concerns the binary, no-covariate, no-nugget specialization of Eqs. \eqref{eq:sgp_binary_obs}--\eqref{eq:difficulty_decomposition}:
\begin{equation}
  Y_{ij}\mid\theta_i,\alpha_j,f
  \sim
  \operatorname{Bernoulli}\{\sigma(\eta_{ij})\},
  \quad
  \eta_{ij}
  =
  \alpha_j\theta_i-f(\mathbf s_j),
  \label{eq:theory_binary_model}
\end{equation}
independently over \(i=1,\ldots,I\) and \(j=1,\ldots,J\), where \(\sigma(x)=\{1+\exp(-x)\}^{-1}\).  The true parameter is denoted by $\eta_{ij}^{0}= \alpha_j^{0}\theta_i^{0}-f_0(\mathbf s_j).$ Let \(n=IJ\), and define the empirical item norm and prediction metric by $\|g\|_J^2 = J^{-1}\sum_{j=1}^{J}g(\mathbf s_j)^2,$ and $d_n^2(\boldsymbol\eta,\boldsymbol\eta^0)= {(IJ)}^{-1}\sum_{i=1}^{I}\sum_{j=1}^{J} (\eta_{ij}-\eta_{ij}^{0})^2$ respectively. The empirical norm is the primary loss because it measures recovery of difficulty at the observed item locations.

\subsubsection*{Assumptions}

\begin{itemize}

\item[(A1)] \textit{Fixed quasi-uniform item geometry.}
The fixed coordinates \(\mathbf s_1,\ldots,\mathbf s_J\in\mathcal S=[0,1]^D\) have fill distance and separation radius
\[
  h_J
  =
  \sup_{\mathbf s\in\mathcal S}
  \min_{1\le j\le J}\|\mathbf s-\mathbf s_j\|,
  \quad
  q_J
  =
  \frac12\min_{j\ne k}\|\mathbf s_j-\mathbf s_k\|,
\]
and, for constants \(0<c<C<\infty\), $h_J\le C J^{-1/D},$ $q_J\ge c J^{-1/D}.$

\item[(A2)] \textit{Smooth truth.}
For \(\nu>0\) and \(B<\infty\), the true surface satisfies $f_0\in\mathcal C^\nu(\mathcal S;B),$ where, with \(r=\lfloor\nu\rfloor\),
\[
  \mathcal C^\nu(\mathcal S;B)
  =
  \left\{
  f:
  \max_{|\boldsymbol\kappa|\le r}
  \|\partial^{\boldsymbol\kappa}f\|_\infty\le B,\quad
  \max_{|\boldsymbol\kappa|=r}
  \sup_{\mathbf s\ne\mathbf t}
  \frac{
    |\partial^{\boldsymbol\kappa}f(\mathbf s)
     -\partial^{\boldsymbol\kappa}f(\mathbf t)|
  }{
    \|\mathbf s-\mathbf t\|^{\nu-r}
  }
  \le B
  \right\}.
\]

\item[(A3)] \textit{Identifiability.}
Both the candidate and true ability vectors obey
\[
  I^{-1}\sum_{i=1}^{I}\theta_i
  =
  I^{-1}\sum_{i=1}^{I}\theta_i^0
  =0,
  \quad
  I^{-1}\sum_{i=1}^{I}\theta_i^2
  =
  I^{-1}\sum_{i=1}^{I}(\theta_i^0)^2
  =1.
\]
For fixed \(0<\underline a<\overline a<\infty\), $\alpha_j,\alpha_j^0\in[\underline a,\overline a],$ $j=1,\ldots,J.$

\item[(A4)] \textit{Uniform likelihood regularity.}
For a finite \(L_0\), \(\max_{i,j}|\eta_{ij}^{0}|\le L_0\).  In addition, for the testing-rate sequence \(\varepsilon_{I,J}\) defined in Appendix~\ref{app:proof-notation}, the theoretical prior admits measurable regularity sieves \(\mathcal R_{I,J}\) containing the true parameter and a fixed \(L<\infty\) such that
\[
  \sup_{\boldsymbol\vartheta\in\mathcal R_{I,J}}
  \max_{i,j}|\eta_{ij}(\boldsymbol\vartheta)|\le L,
  \quad
  \Pi(\mathcal R_{I,J}^{c})
  \le
  \exp\{-C_{\mathcal R}IJ
  \varepsilon_{I,J}^{\,2}\},
\]
where \(C_{\mathcal R}\) can be chosen as an arbitrarily large fixed constant.
\item[(A5)] \textit{Anisotropic GP and nuisance priors.}
Let the raw hierarchical GP block be represented as
\[
  f(\mathbf s)
  =
  \sigma_f W(\mathbf A\mathbf s),
  \quad
  \mathbf A=\operatorname{diag}(A_1,\ldots,A_D),
\]
where \(W\) is a centered squared-exponential Gaussian process with unit marginal variance and, independently, \(A_d^{D} \sim\operatorname{Gamma}(a_A,b_A)\), \(d=1,\ldots,D\). Additionally, we require:

\begin{itemize}
\item[(i)] The amplitude \(\sigma_f\) is either fixed and positive or has a density
bounded below near one and satisfies
\[
\Pi(\sigma_f>t)\le C_\sigma e^{-c_\sigma t^2}
\]
for all sufficiently large \(t\).

\item[(ii)] The GP law
is conditioned on \(\|f\|_{\infty} \le B_f\), for some \(B_f>B\), and the ability prior is restricted to
\[
\max_i|\theta_i|\le B_{\theta},
\quad
B_\theta>\max\left\{1,\frac{L_0+B}{\underline a}\right\}.
\]
\item[(iii)] The nuisance priors are supported on the identified spaces in
(A3), and their densities, with respect to the corresponding surface and Lebesgue measures, are bounded below on fixed neighborhoods of the truth by \(e^{-c_\theta I}\) and \(e^{-c_\alpha J}\), respectively, for constants \(c_\theta,c_\alpha<\infty\) independent of \(I\) and \(J\).
\end{itemize}

\item[(A6)] \textit{Joint growth.}
The numbers of respondents and items satisfy \(\min(I,J)/\log(IJ)\to\infty\).
\end{itemize}

These assumptions bridge the finite-dimensional identifiability requirements of IRT and the infinite-dimensional recovery conditions of Bayesian nonparametrics. The quasi-uniform design in (A1) avoids a regular-grid requirement while ensuring that the observation locations become dense without pathological clustering, frequently required for continuum-\(L^2\) contraction results \parencite{vandervaartvanzanten_2008, bhattacharya2015posteriorcontractiongaussianprocess}. The Hölder smoothness condition in (A2), together with the hierarchical squared-exponential prior in (A5), provides the approximation, local-mass, and tail conditions needed for posterior contraction \parencite{vaartzanten2009}. Assumption (A3) resolves the location, scale, and sign indeterminacies of the 2PL model \parencite{Fox2010}, whereas (A4) supplies uniform local quadratic likelihood geometry on a regularity sieve carrying all prior mass asymptotically \parencite{Ghosal_2007}. Finally, the joint growth condition in (A6) addresses the incidental parameter problem: requiring \(\min(I,J)/\log(IJ)\to\infty\) allows unbalanced panel growth while making the entropy contribution from the growing nuisance vectors \(\boldsymbol\theta\) and \(\boldsymbol\alpha\) asymptotically negligible relative to the spatial recovery rate \parencite{Haberman1977, Chen2019}.

The restrictions in (A5) make the role of the regularity sieve explicit. Because \(\|f_0\|_{\infty}\le B\) and \(\max_{i,j} |\eta_{ij}^0|\le L_0\), assumption (A3) implies \( \max_i |\theta_i^0| \le {(L_0+B)}/{\underline a} < B_\theta.\) Thus, the true parameter lies in the interior of the truncated theoretical support. Moreover, every parameter in that support satisfies \( |\alpha_j\theta_i-f(\mathbf s_j)| \le \overline a B_\theta+B_f=:L. \) Consequently, (A4) holds with \(\mathcal R_{I,J}\) equal to the full prior support and with zero prior mass outside the sieve. The stated lower bounds on the nuisance densities similarly ensure sufficient prior mass around \(\boldsymbol\theta_0\) and \(\boldsymbol\alpha_0\). These fixed truncations are theoretical regularity conditions, and the contraction result is therefore stated for the corresponding bounded prior.

\subsubsection*{Posterior contraction under the IRT likelihood}

We define the rate components $\varepsilon_{I,J}^{\mathrm{GP}}= (IJ)^{-\nu/(2\nu+D)} \{\log(IJ)\}^{\kappa},$ $\varepsilon_{I,J}^{\mathrm{nuis}} = \left\{ {(I+J)\log(IJ)}/{IJ} \right\}^{1/2},$ and $\varepsilon_{I,J}^{\mathrm{item}} = \max\left\{ \varepsilon_{I,J}^{\mathrm{GP}}, \varepsilon_{I,J}^{\mathrm{nuis}} \right\},$ where \(\kappa>0\) is a sufficiently large constant determined by the squared-exponential GP concentration bound, \(\nu\), and \(D\).

\begin{theorem}[Difficulty-surface contraction]
\label{thm:sgp_irt_contraction}
Suppose that (A1)--(A6) hold and that \(IJ(\varepsilon_{I,J}^{\mathrm{item}})^2\to\infty\).  Then there exists a sufficiently large constant \(M>0\) such that
\[
  \Pi \left(
    \|f-f_0\|_J
    >
    M\varepsilon_{I,J}^{\mathrm{item}}
    \,\middle|\,
    \mathbf Y
  \right)
  \longrightarrow0
\]
in \(P_0\)-probability as \(I,J\to\infty\) jointly.
\end{theorem}

The two terms in $\varepsilon_{I,J}^{\mathrm{item}}$ have distinct interpretations. The first term $\varepsilon_{I,J}^{\mathrm{GP}}$ is the smooth-function contribution from the hierarchical GP prior under \(IJ\) binary observations.  The logarithmic factor is the adaptation cost associated with a squared-exponential process with unknown rescaling parameters.  The second term $\varepsilon_{I,J}^{\mathrm{nuis}}$ is the price of estimating the \(I\)-dimensional ability vector and the \(J\)-dimensional discrimination vector jointly with \(f\).  It disappears from an oracle calculation that incorrectly conditions on the true nuisance parameters.  Assumption (A6) is what makes this nuisance contribution vanish, since a limit with only \(I\) or only \(J\) increasing does not yield the stated result for the full model.

\begin{remark}\label{rem:l2-rate}
For recovery over the entire spatial domain rather than only at the observed items, suppose additionally that the quasi-uniform design admits a stable order-\(\nu\) reconstruction operator and that the posterior is restricted to the corresponding H\"older sieve.  Then the discretization error is of order \(h_J^\nu\asymp J^{-\nu/D}\).  Defining $\varepsilon_{I,J}^{\mathrm{cont}}=\max\left\{ \varepsilon_{I,J}^{\mathrm{item}}, J^{-\nu/D} \right\},$ the same conclusion holds with \(\|f-f_0\|_{L^2(\mathcal S)}\) and \(\varepsilon_{I,J}^{\mathrm{cont}}\) in place of \(\|f-f_0\|_J\) and \(\varepsilon_{I,J}^{\mathrm{item}}\).
\end{remark}

The proof of Theorem \ref{thm:sgp_irt_contraction} applies the exact proof architecture of \textcite{Ghosal_2007} for independent, non-identically distributed Bernoulli observations directly to Eq.\eqref{eq:theory_binary_model}.  On the sieves in (A4), Bernoulli Kullback--Leibler neighborhoods and tests are controlled by \(d_n(\boldsymbol\eta,\boldsymbol\eta^0)\).  The RKHS approximation, small-ball, and metric-entropy results for the rescaled GP in (A5) verify the prior concentration and sieve conditions for \(f\) \parencite{vandervaartvanzanten_2008,vandervaart2011}; the nuisance parameters contribute to the metric entropy of the sieve.  Finally, the exact centering constraints give, for every item \(j\),
\begin{align}
  \frac1I\sum_{i=1}^{I}
  (\eta_{ij}-\eta_{ij}^{0})
  &=
  \alpha_j\left(\frac1I\sum_{i=1}^{I}\theta_i\right)
  -
  \alpha_j^0\left(\frac1I\sum_{i=1}^{I}\theta_i^0\right)
  -
  \{f(\mathbf s_j)-f_0(\mathbf s_j)\}
  \notag\\
  &=
  -\{f(\mathbf s_j)-f_0(\mathbf s_j)\}.
  \label{eq:prediction_to_surface}
\end{align}
Jensen's inequality therefore yields \( \|f-f_0\|_J^2 \le d_n^2(\boldsymbol\eta,\boldsymbol\eta^0), \) which transfers contraction of the full linear predictor to contraction of the difficulty surface without substituting estimates of \(\boldsymbol\theta\) or \(\boldsymbol\alpha\) for their true values.

\begin{remark} \label{rem:theory_scope}

Theorem \ref{thm:sgp_irt_contraction} deliberately concerns the no-nugget submodel \(u_j\equiv0\).  With a nonzero nugget, \(\beta_j=f(\mathbf s_j)+u_j\), there are two different inferential targets. Repeated responses to item \(j\) increase information about its realized difficulty \(\beta_j\), but they cannot average away the latent item-specific deviation \(u_j\) when the target is the smooth surface \(f\).  Smooth-surface recovery then resembles nonparametric regression from \(J\) noisy item-level quantities and is generally governed by the number of item locations, rather than by \(IJ\) alone.  In particular, the no-nugget rate in Theorem \ref{thm:sgp_irt_contraction} should not be quoted as a formal contraction rate for \(f\) when \(\sigma_u^2\) remains bounded away from zero.

Thus, under fixed-domain asymptotics, marginal variance and range parameters need not be separately consistently estimable, even when prediction of the latent surface is consistent \parencite{zhang2004}. Results for GP models with a nugget likewise distinguish surface recovery from recovery of individual covariance parameters \parencite{li2024}.  We therefore use the posterior for \(\sigma_f,\boldsymbol\ell,\sigma_u\) to propagate regularization uncertainty and assess finite-sample calibration; Theorem \ref{thm:sgp_irt_contraction} makes no covariance-parameter consistency claim.
\end{remark}

\subsubsection*{Fixed-category polytomous extension}
\label{subsec:poly_theory}

The direct likelihood argument also applies to the identified softmax model in Eqs. \eqref{eq:sgp_poly_predictor} and \eqref{eq:sgp_poly_reference} when \(C\) is fixed.  For clarity, we omit covariates and the nugget and write \( \eta_{ijc} = q_c\alpha_j\theta_i-f_c(\mathbf s_j), \ c=1,\ldots,C-1. \)

\begin{corollary}[Polytomous contraction]
\label{cor:poly_contraction}
Suppose that \(C\) and \(q_1,\ldots,q_{C-1}\) are fixed, each \(q_c\ne0\), and assumptions analogous to (A1)--(A6) hold jointly for \(f_{01},\ldots,f_{0,C-1}\in\mathcal C^\nu(\mathcal S;B)\).  Then, for a sufficiently large \(M>0\),
\[
  \Pi\!\left(
    \sum_{c=1}^{C-1}
    \|f_c-f_{0c}\|_J^2
    >
    M^2(\varepsilon_{I,J}^{\mathrm{item}})^2
    \,\middle|\,
    \mathbf Y
  \right)
  \longrightarrow0
\]
in \(P_0\)-probability.  Because \(C\) is fixed, the constants may depend on \(C\), but the order of \(\varepsilon_{I,J}^{\mathrm{item}}\) is unchanged.
\end{corollary}

Corollary \ref{cor:poly_contraction} shows that the recovery guarantee is not limited to dichotomous responses. For any fixed response scale, the posterior jointly recovers the entire collection of category-specific difficulty surfaces at the same item-level contraction rate order as in the binary model. The reason is that, on bounded-logit sets, the multinomial Kullback--Leibler divergence and Hellinger distance are uniformly equivalent to squared Euclidean distance between the vectors of nonreference logits. Because \(C\) is fixed, introducing \(C-1\) GP surfaces changes the prior concentration and entropy bounds only by a constant factor. Thus, the theoretical benefit of spatial smoothing is retained for polytomous assessments without an additional rate order penalty. The result also accommodates fixed-dimensional category-specific covariate effects, whose parametric complexity is absorbed by the stated rate.

\section{Simulation Study}
\label{sec:simulation}

We designed the simulation study to evaluate finite-sample recovery, uncertainty quantification, and out-of-sample prediction at a similar item count used in the empirical application. The design follows the ecological setting of Figure \ref{fig:ecological_motivation}, in which items sit at observed coordinates on a bounded region.  Through these experiments, we try to address: (i) whether the anisotropic kernel can exploit direction-dependent dependence; (ii) whether the smooth and item-specific components can be distinguished when the nugget is non-negligible; (iii) whether the model remains competitive under rough, discontinuous, or CAR-generated difficulty surfaces; and (iv) whether spatial smoothing causes deterioration when the truth contains no spatial signal.

\paragraph{Common observation model.}
All six scenarios use the identified binary 2PL observation model in Eqs. \eqref{eq:sgp_binary_obs}--\eqref{eq:sgp_binary_predictor}.  We generate $I=200$ respondents and $J=50$ items, yielding $10{,}000$ response cells. Item coordinates are generated from a jittered quasi-uniform design on the unit square.  We first draw independent standard-normal respondent scores and then center and scale the realized vector so that \( I^{-1}\sum_{i=1}^I\theta_i=0, \ I^{-1}\sum_{i=1}^I\theta_i^2=1. \) Item discriminations are generated as $\alpha_j=\exp(a_j)$ with $a_j\stackrel{\mathrm{iid}}{\sim} \mathcal{N}(0,0.35^2)$, which ensures $\alpha_j>0$.  A single observation-level covariate is drawn as $x_{ij}\sim\mathrm{Unif}(-1,1)$ and assigned coefficient $\gamma=0.5$.  Conditional on these quantities,
\[
  Y_{ij}\sim\mathrm{Bernoulli}(p_{ij}),
  \quad
  \operatorname{logit}(p_{ij})
  =\alpha_j\theta_i-\beta_j+x_{ij}\gamma.
\]

The difficulty vector is decomposed as $\boldsymbol\beta=\mathbf f+\mathbf u$, where $\mathbf f$ is the structured component and $\mathbf u\sim \mathcal{N}(\mathbf 0,\sigma_u^2\mathbf I_J)$ is item-specific noise. The six scenarios modify only the distribution of $\mathbf f$ and the size of $\sigma_u$.

\paragraph{Difficulty scenarios.}
The following cases include both favorable and deliberately misspecified settings for the squared-exponential SGP prior.

\begin{enumerate}
\item \textbf{Anisotropic GP, low nugget.}
We generate $\mathbf f$ from an anisotropic squared-exponential GP with $\sigma_f=1$, length scales $(\ell_1,\ell_2)=(0.15,0.45)$, and $\sigma_u=0.15$.  The directional length scales differ by a factor of three, so this setting directly assesses anisotropy rather than merely smoothness. Its nugget-to-total-variance ratio is approximately $0.022$.

  \item \textbf{Anisotropic GP, high nugget.}
The length scales remain $(0.15,0.45)$, but we set $\sigma_f=0.65$ and $\sigma_u=0.90$, giving a nugget-to-total-variance ratio of approximately $0.657$. This scenario tests whether the fitted model can avoid attributing unstructured heterogeneity to a short spatial range.

  \item \textbf{CAR truth.}
Let $W$ be the Delaunay adjacency matrix of the irregular item coordinates, $D=\operatorname{diag}(W\mathbf 1)$, and $Q=D-0.95W$.  We generate
  \[
    \mathbf f\sim
    \mathcal{N} \left(\mathbf 0,\,0.35^2 Q^{-1}\right),
    \quad
    \mathbf u\sim \mathcal{N}(\mathbf 0,0.15^2\mathbf I_J).
  \]
This is a CAR-favorable scenario and checks whether \texttt{SGP-IRT} remains competitive when the structured component follows the local Markov dependence assumed by the PLUS competitors.

  \item \textbf{Rough Mat\'ern truth.}
We generate $\mathbf f$ from an anisotropic Mat\'ern process with smoothness $\nu=0.5$ (the same symbol denotes the H\"older smoothness of the truth in Assumption (A2)), $(\ell_1,\ell_2)=(0.15,0.35)$, $\sigma_f=1$, and $\sigma_u=0.15$.  The resulting exponential covariance produces sample paths rougher than those implied by the fitted squared-exponential kernel.

  \item \textbf{Piecewise-smooth truth.}
We add a discontinuous regional jump to a smooth anisotropic surface:
  \[
    f_j=f_j^{\mathrm{SE}}
      +\left\{\mathbf{1}(s_{j1}>0.5)
      -J^{-1}\sum_{k=1}^J\mathbf{1}(s_{k1}>0.5)\right\},
  \]
where $\mathbf f^{\mathrm{SE}}$ has $\sigma_f=0.60$ and $(\ell_1,\ell_2)=(0.18,0.35)$, and $\sigma_u=0.10$.  The unit jump makes this scenario deliberately adversarial to a globally smooth kernel.

  \item \textbf{IID, no spatial signal.}
We set $\mathbf f=\mathbf 0$ and generate $\mathbf u\sim \mathcal{N}(\mathbf 0,\mathbf I_J)$.  This negative control tests whether the GP spuriously smooths an entirely unstructured difficulty vector.
\end{enumerate}

Figure \ref{fig:sim_scenarios} illustrates all six difficulty scenarios.

\begin{figure}[H]
  \centering
  \includegraphics[width=\linewidth]{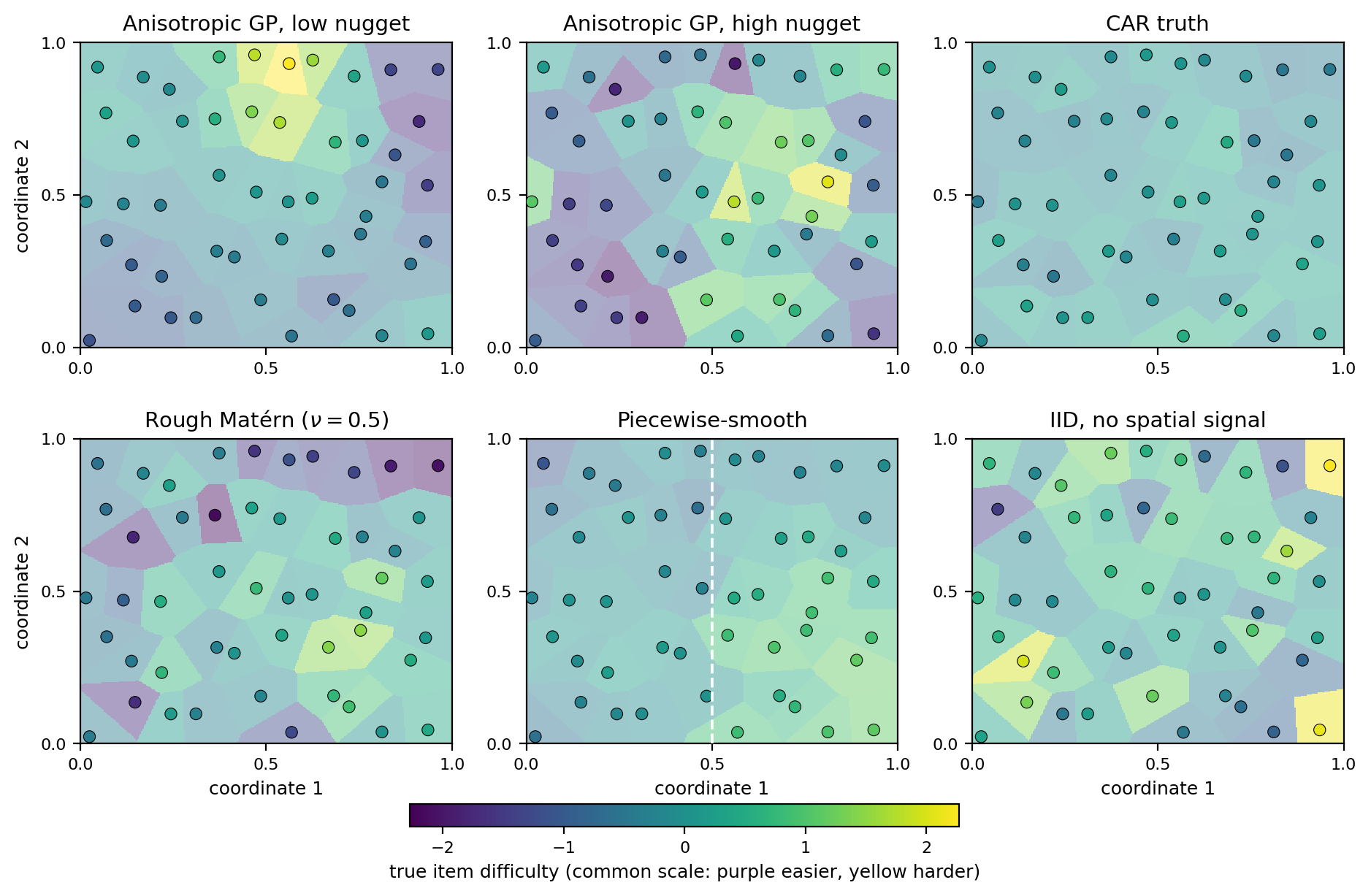}
  \caption{Example true item-difficulty surfaces for the six simulation
  scenarios.  Each panel shows the $J=50$ irregularly located items on the unit
  square, with light shading indicating nearest-neighbor regions.  All panels
  share one color scale, so amplitudes are comparable across scenarios.  Raising
  the nugget makes neighboring difficulties differ more, and the CAR surface
  varies over a narrower range than the others.  The rough Mat\'ern surface
  varies sharply between nearby items, the piecewise-smooth surface jumps across
  the boundary at the midpoint of the first coordinate (dashed line), and the
  final panel has no spatial structure.}
  \label{fig:sim_scenarios}
\end{figure}

\subsection*{Competing methods and evaluation}

For each generated dataset, we fit \texttt{SGP-IRT} and the CAR-based \texttt{1PLUS}, \texttt{2PLUS}, and \texttt{3PLUS} procedures to the same training responses and evaluate them on the same held-out cells.  The holdout is stratified within respondent: ten of each respondent's 50 responses are withheld, producing $2{,}000$ test cells while retaining every respondent and item in the training data.

The PLUS fits are not merely likelihood-matched ablations of \texttt{SGP-IRT}.  In particular, their 2PL linear predictor is parameterized as $\alpha_j(\theta_i-b_j-\delta_g)$, written here in the notation of Section~\ref{sec:background}, so their difficulty lies on the ability scale.  The term $\delta_g$ is an effect for group $g$, since they supply the covariate information through a prespecified three-level grouping of $x_{ij}$.  \texttt{SGP-IRT} uses the continuous term $x_{ij}\gamma$.  To make parameter recovery meaningful, we transform every PLUS posterior draw to the identification $\mathrm{mean}(\theta)=0$ and $\mathrm{mean}(\theta^2)=1$, and convert the raw PLUS difficulty to the corresponding intercept-scale difficulty $\beta_j$ before computing errors.  The \texttt{1PLUS} model fixes a common discrimination and hence does not yield item-specific $\alpha_j$ recovery.  The \texttt{3PLUS} model estimates a pseudo-guessing component even though the data-generating mechanisms contain no guessing and is therefore an intentionally over-parameterized sensitivity comparison.

We assess recovery of $\theta$, $\alpha$, and $\beta$ using RMSE of the posterior means.  To evaluate uncertainty rather than point estimation alone, we also report the proportion of true parameter values contained in 95\% shortest posterior intervals and the corresponding average interval width. Out-of-sample performance is summarized by negative mean log predictive density (NLPD), Brier score, AUC, classification accuracy at threshold $0.5$, and RMSE between the posterior predictive probability and the true generating probability.  Lower NLPD, Brier score, and probability RMSE are preferable; higher AUC and accuracy are preferable. Each scenario uses 25 independent replications, with every method fitted to the same generated datasets and evaluated on the same held-out cells within a replication. Each fit uses three chains of 1,500 iterations, of which 500 are discarded as warmup. These settings are lighter than those in Section \ref{sec:method} because the study requires 600 fits.

\begin{table}[htbp]
\centering
\caption{Parameter recovery and uncertainty quantification for \texttt{SGP-IRT} and the
\texttt{PLUS} models across six simulation scenarios. Entries are averages across 25 independent replications. Boldface denotes the smallest RMSE, coverage closest to the nominal 0.95 level, and smallest interval width within each scenario.}
\label{tab:four_method_parameter}
\scriptsize
\setlength{\tabcolsep}{3.5pt}
\resizebox{\textwidth}{!}{%
\begin{tabular}{llrrrrrrrrr}
\toprule
& & \multicolumn{3}{c}{$\theta$} & \multicolumn{3}{c}{$\alpha$}
& \multicolumn{3}{c}{$\beta$} \\
\cmidrule(lr){3-5}\cmidrule(lr){6-8}\cmidrule(lr){9-11}
\textbf{Scenario} & \textbf{Method}
& \textbf{RMSE} & \textbf{Coverage} & \textbf{Width}
& \textbf{RMSE} & \textbf{Coverage} & \textbf{Width}
& \textbf{RMSE} & \textbf{Coverage} & \textbf{Width} \\
\midrule
\multirow{4}{*}{Anisotropic GP, low nugget}
& \texttt{SGP-IRT} & \textbf{0.324} & \textbf{0.955} & 1.263 & \textbf{0.223} & \textbf{0.950} & 0.770 & \textbf{0.150} & 0.910 & \textbf{0.542} \\
& \texttt{1PLUS} & 0.341 & 0.908 & 1.190 & -- & -- & -- & 0.169 & 0.910 & 0.606 \\
& \texttt{2PLUS} & 0.325 & 0.903 & \textbf{1.162} & 0.279 & 0.750 & \textbf{0.643} & 0.168 & \textbf{0.930} & 0.622 \\
& \texttt{3PLUS} & 0.331 & 0.903 & 1.185 & 0.305 & 0.790 & 0.809 & 0.267 & 0.840 & 0.773 \\
\addlinespace
\multirow{4}{*}{Anisotropic GP, high nugget}
& \texttt{SGP-IRT} & \textbf{0.357} & 0.940 & 1.346 & \textbf{0.212} & \textbf{0.950} & 0.747 & \textbf{0.206} & \textbf{0.940} & 0.764 \\
& \texttt{1PLUS} & 0.369 & 0.923 & 1.317 & -- & -- & -- & 0.224 & 0.900 & \textbf{0.700} \\
& \texttt{2PLUS} & \textbf{0.357} & 0.920 & \textbf{1.298} & 0.252 & 0.800 & \textbf{0.628} & 0.221 & \textbf{0.940} & 0.736 \\
& \texttt{3PLUS} & \textbf{0.357} & \textbf{0.943} & 1.305 & 0.259 & 0.840 & 0.737 & 0.298 & 0.890 & 0.860 \\
\addlinespace
\multirow{4}{*}{CAR truth}
& \texttt{SGP-IRT} & \textbf{0.327} & \textbf{0.950} & 1.250 & \textbf{0.201} & \textbf{0.960} & 0.766 & \textbf{0.130} & 0.970 & 0.570 \\
& \texttt{1PLUS} & 0.333 & 0.935 & 1.244 & -- & -- & -- & 0.147 & \textbf{0.950} & \textbf{0.554} \\
& \texttt{2PLUS} & 0.329 & 0.948 & \textbf{1.219} & 0.215 & 0.890 & \textbf{0.679} & 0.140 & 0.910 & 0.562 \\
& \texttt{3PLUS} & 0.330 & 0.925 & 1.251 & 0.268 & \textbf{0.940} & 0.907 & 0.278 & 0.770 & 0.729 \\
\addlinespace
\multirow{4}{*}{Rough Mat\'ern}
& \texttt{SGP-IRT} & \textbf{0.335} & \textbf{0.930} & 1.214 & \textbf{0.224} & \textbf{0.950} & 0.889 & \textbf{0.189} & 0.940 & 0.753 \\
& \texttt{1PLUS} & 0.345 & 0.903 & 1.179 & -- & -- & -- & 0.207 & 0.890 & \textbf{0.676} \\
& \texttt{2PLUS} & 0.340 & 0.883 & \textbf{1.139} & 0.290 & 0.850 & \textbf{0.779} & 0.213 & \textbf{0.950} & 0.743 \\
& \texttt{3PLUS} & 0.338 & 0.910 & 1.153 & 0.308 & 0.850 & 0.868 & 0.270 & 0.920 & 0.843 \\
\addlinespace
\multirow{4}{*}{Piecewise-smooth}
& \texttt{SGP-IRT} & 0.337 & \textbf{0.935} & 1.287 & \textbf{0.236} & \textbf{0.930} & 0.774 & \textbf{0.153} & \textbf{0.970} & 0.619 \\
& \texttt{1PLUS} & 0.368 & 0.878 & 1.190 & -- & -- & -- & 0.193 & 0.890 & \textbf{0.589} \\
& \texttt{2PLUS} & 0.351 & 0.870 & \textbf{1.153} & 0.278 & 0.810 & \textbf{0.628} & 0.190 & 0.840 & 0.638 \\
& \texttt{3PLUS} & \textbf{0.336} & 0.888 & 1.195 & 0.284 & 0.850 & 0.712 & 0.219 & 0.880 & 0.706 \\
\addlinespace
\multirow{4}{*}{IID, no spatial signal}
& \texttt{SGP-IRT} & \textbf{0.312} & 0.945 & 1.256 & \textbf{0.241} & \textbf{0.940} & 0.795 & \textbf{0.192} & \textbf{0.960} & 0.743 \\
& \texttt{1PLUS} & 0.317 & 0.930 & 1.245 & -- & -- & -- & 0.203 & 0.910 & \textbf{0.706} \\
& \texttt{2PLUS} & 0.317 & \textbf{0.948} & \textbf{1.208} & 0.259 & 0.810 & \textbf{0.669} & 0.193 & \textbf{0.960} & 0.739 \\
& \texttt{3PLUS} & 0.315 & \textbf{0.948} & \textbf{1.208} & 0.298 & 0.760 & 0.749 & 0.271 & 0.880 & 0.817 \\
\bottomrule
\end{tabular}%
}
\end{table}

\begin{table}[htbp]
\centering
\caption{Held-out predictive performance of \texttt{SGP-IRT} and the PLUS models across six simulation scenarios. Lower negative mean log predictive density (NLPD),
Brier score, and probability RMSE are better; higher AUC and accuracy are
better. Entries are averages across 25 independent replications. Boldface identifies the best value within each scenario and metric.}
\label{tab:four_method_prediction}
\scriptsize
\setlength{\tabcolsep}{5pt}
\resizebox{\textwidth}{!}{%
\begin{tabular}{llrrrrr}
\toprule
\textbf{Scenario} & \textbf{Method} & \textbf{NLPD}
& \textbf{Brier} & \textbf{AUC} & \textbf{Accuracy}
& \textbf{Probability RMSE} \\
\midrule
\multirow{4}{*}{Anisotropic GP, low nugget}
& \texttt{SGP-IRT} & \textbf{0.5700} & \textbf{0.1940} & \textbf{0.7640} & \textbf{0.7000} & \textbf{0.0750} \\
& \texttt{1PLUS} & 0.5798 & 0.1981 & 0.7572 & 0.6963 & 0.0910 \\
& \texttt{2PLUS} & 0.5749 & 0.1962 & 0.7622 & 0.6945 & 0.0828 \\
& \texttt{3PLUS} & 0.5753 & 0.1965 & 0.7615 & 0.6927 & 0.0822 \\
\addlinespace
\multirow{4}{*}{Anisotropic GP, high nugget}
& \texttt{SGP-IRT} & \textbf{0.5290} & \textbf{0.1770} & \textbf{0.8060} & 0.7260 & \textbf{0.0800} \\
& \texttt{1PLUS} & 0.5364 & 0.1799 & 0.8026 & 0.7285 & 0.0900 \\
& \texttt{2PLUS} & 0.5324 & 0.1785 & 0.8055 & 0.7263 & 0.0855 \\
& \texttt{3PLUS} & 0.5324 & 0.1784 & \textbf{0.8060} & \textbf{0.7288} & 0.0851 \\
\addlinespace
\multirow{4}{*}{CAR truth}
& \texttt{SGP-IRT} & \textbf{0.5960} & \textbf{0.2050} & \textbf{0.7420} & \textbf{0.6790} & \textbf{0.0770} \\
& \texttt{1PLUS} & 0.6037 & 0.2083 & 0.7358 & 0.6723 & 0.0903 \\
& \texttt{2PLUS} & 0.5987 & 0.2066 & 0.7397 & 0.6750 & 0.0827 \\
& \texttt{3PLUS} & 0.5991 & 0.2067 & 0.7393 & 0.6760 & 0.0837 \\
\addlinespace
\multirow{4}{*}{Rough Mat\'ern}
& \texttt{SGP-IRT} & \textbf{0.5480} & \textbf{0.1850} & \textbf{0.7900} & \textbf{0.7220} & \textbf{0.0800} \\
& \texttt{1PLUS} & 0.5549 & 0.1878 & 0.7852 & 0.7188 & 0.0947 \\
& \texttt{2PLUS} & 0.5529 & 0.1872 & 0.7864 & 0.7200 & 0.0869 \\
& \texttt{3PLUS} & 0.5513 & 0.1866 & 0.7878 & 0.7175 & 0.0852 \\
\addlinespace
\multirow{4}{*}{Piecewise-smooth}
& \texttt{SGP-IRT} & \textbf{0.5760} & \textbf{0.1970} & \textbf{0.7570} & \textbf{0.6915} & \textbf{0.0811} \\
& \texttt{1PLUS} & 0.5890 & 0.2018 & 0.7472 & 0.6845 & 0.0975 \\
& \texttt{2PLUS} & 0.5814 & 0.1991 & 0.7532 & 0.6907 & 0.0894 \\
& \texttt{3PLUS} & 0.5798 & 0.1985 & 0.7547 & 0.6872 & 0.0870 \\
\addlinespace
\multirow{4}{*}{IID, no spatial signal}
& \texttt{SGP-IRT} & \textbf{0.5500} & \textbf{0.1870} & 0.7880 & \textbf{0.7140} & \textbf{0.0767} \\
& \texttt{1PLUS} & 0.5598 & 0.1889 & 0.7859 & 0.7110 & 0.0855 \\
& \texttt{2PLUS} & 0.5566 & 0.1880 & \textbf{0.7882} & 0.7103 & 0.0821 \\
& \texttt{3PLUS} & 0.5577 & 0.1883 & 0.7876 & 0.7110 & 0.0815 \\
\bottomrule
\end{tabular}%
}

\end{table}

\subsection*{Results}

\paragraph{Parameter recovery and uncertainty quantification.}
Table \ref{tab:four_method_parameter} shows that the ability parameter recovery is similar across the four procedures.  Averaged across scenarios, the $\theta$ RMSE is $0.332$ for \texttt{SGP-IRT}, compared with $0.345$, $0.337$, and $0.335$ for \texttt{1PLUS}, \texttt{2PLUS}, and \texttt{3PLUS}, respectively.  The only scenario in which \texttt{SGP-IRT} does not attain the smallest displayed ability RMSE is the piecewise case, where \texttt{3PLUS} gives $0.336$ and \texttt{SGP-IRT} gives $0.337$.  This difference is negligible at the precision of the experiment and indicates that the main benefit of spatial modeling is not obtained by materially changing the respondent scale.

The methods separate more clearly for item discrimination, which carries the same meaning in both parameterizations and so needs no rescaling before comparison.  Across scenarios, \texttt{SGP-IRT} reduces $\alpha$ RMSE from $0.262$ for the stronger applicable PLUS competitor, \texttt{2PLUS}, to $0.223$.  Its average coverage is also $0.947$, close to the nominal $0.95$, whereas \texttt{2PLUS} and \texttt{3PLUS} attain only $0.818$ and $0.838$.  The narrower discrimination intervals produced by \texttt{2PLUS} therefore do not represent uniformly greater precision as they accompany marked undercoverage.  This result is particularly relevant to discrimination recovery because it evaluates both point accuracy and the posterior uncertainty induced by the $\theta$--$\alpha$ dependence.

Comparison on the difficulty scale requires the PLUS draws to be converted to the intercept scale, as described above.  \texttt{SGP-IRT} then has the smallest $\beta$ RMSE in all six scenarios.  Its across-scenario average is $0.170$, compared with $0.190$, $0.188$, and $0.267$ for \texttt{1PLUS}, \texttt{2PLUS}, and \texttt{3PLUS}.  The gain is largest under the piecewise truth ($0.153$ versus $0.190$ for the best PLUS result) but is also present under rough Mat\'ern truth ($0.189$ versus $0.207$), showing that the fitted squared-exponential prior is not dependent on exact kernel specification to remain competitive.  Under CAR truth, \texttt{SGP-IRT} obtains $0.130$, compared with $0.140$ for \texttt{2PLUS}; hence its advantage is not restricted to data generated from a GP.  Conversely, under the i.i.d. negative control, \texttt{SGP-IRT} and \texttt{2PLUS} are effectively tied ($0.192$ and $0.193$), indicating little cost from including the spatial component when the truth is unstructured.

The high-nugget and low-nugget cases clarify the role of item-specific noise. For low nugget, \texttt{SGP-IRT} gives $\beta$ RMSE $0.150$ and interval width $0.542$; for high nugget, these increase to $0.206$ and $0.764$.  Thus recovery becomes less accurate when most variation is item specific, as expected, but coverage remains $0.940$.  Across all scenarios, \texttt{SGP-IRT}'s $\beta$ coverage is $0.948$ with average width $0.665$, compared with coverage between $0.863$ and $0.922$ for the PLUS procedures. The first of our four questions concerns the length scales themselves. They are recovered adequately under the low-nugget truth and less sharply under the high-nugget truth, as Appendix \ref{app:additional_sim} reports.

\paragraph{Held-out prediction.}
Table \ref{tab:four_method_prediction} shows modest but consistent predictive gains.  Across scenarios, \texttt{SGP-IRT} has the lowest NLPD ($0.5615$), Brier score ($0.1908$), and probability RMSE ($0.0783$), together with the largest AUC ($0.7745$) and accuracy ($0.7054$).  The best corresponding PLUS values are $0.5659$, $0.1925$, $0.0841$, $0.7728$, and $0.7028$, respectively.  The probability-RMSE reduction is the largest, whereas the differences in thresholded accuracy are small.  This is consistent with the fact that improved recovery of the latent difficulty surface need not change many classifications when predicted probabilities already lie on the same side of $0.5$.

Under the high-nugget truth, all three PLUS models attain slightly higher accuracy than \texttt{SGP-IRT} ($0.7263$ to $0.7288$ against $0.7260$), and \texttt{3PLUS} matches its displayed AUC. Under the i.i.d. truth, \texttt{2PLUS} has a marginally larger AUC ($0.7882$ versus $0.7880$).  \texttt{SGP-IRT}'s predictive advantage persists under CAR, rough Mat\'ern, and piecewise truth, although the absolute improvements are moderate rather than dramatic.

\section{Empirical Study}
\label{sec:empirical}

\subsection{Author Recognition Tests (ARTs)}

We apply \texttt{SGP-IRT} to Author Recognition Test data administered to older adults \parencite{Wimmer2023}. Recognition is organized partly by genre, a dependence that standard IRT ignores by treating items as independent \parencite{Moore2015,StratovichART,mcgeown2020}.

We place the ART items in a latent similarity space with MDS and fit \texttt{SGP-IRT} across three genres (\texttt{SciFiFantasy}, \texttt{SocialPoliticalCommentary}, and \texttt{Suspense}). The subsections below examine the recovered difficulty surface, genre-level operating characteristics, and out-of-sample prediction, the last using respondents' standardized continuous vocabulary scores as a covariate.

Our dataset comprises responses from older adults (aged approximately 50--80 years, with a total sample size of 321 respondents) who completed various measures of print exposure. We focus on three subsets of the Author Recognition Test (ART): \texttt{SciFiFantasy} (25 items), \texttt{SocialPoliticalCommentary} (10 items), and \texttt{Suspense} (25 items). Each item captures a binary response (0/1), indicating recognition of an author as a genuine fiction writer. During preprocessing, we converted responses to numeric format, removed missing values, and excluded items with zero variance.

\subsubsection{Spatial Structure of Item Difficulty}
Figure~\ref{fig:art_motivation}, shown in the Introduction as motivation, is the fitted surface from \texttt{SGP-IRT}.  Each item sits at its MDS coordinates, colored by the posterior mean of its difficulty and sized by the posterior standard deviation.  Larger points therefore mark items whose difficulty is estimated with less certainty.  That uncertainty reflects how much response information an item carries and how much support it receives from nearby items in the embedding.  Because the coordinates are estimated rather than observed, the displayed uncertainty is conditional on the fitted embedding.

Within each genre, the easiest and hardest items occupy distinct regions of the embedding rather than being interspersed, so difficulty varies gradually over the estimated geometry.  Figure~\ref{fig:art_variogram} supplements the map with a distance-based assessment of that structure.

The semivariogram in Figure~\ref{fig:art_variogram} makes this pattern quantitative. Semivariance rises with MDS distance, so items close in the embedding have more similar estimated difficulties than distant ones, the spatial dependence the GP prior is meant to capture. The rise is clearest for \texttt{SciFiFantasy} and \texttt{Suspense}; \texttt{SocialPoliticalCommentary} has only ten items, so we read its curve with caution.

\begin{figure}[htbp]
    \centering
    \includegraphics[width=\linewidth]{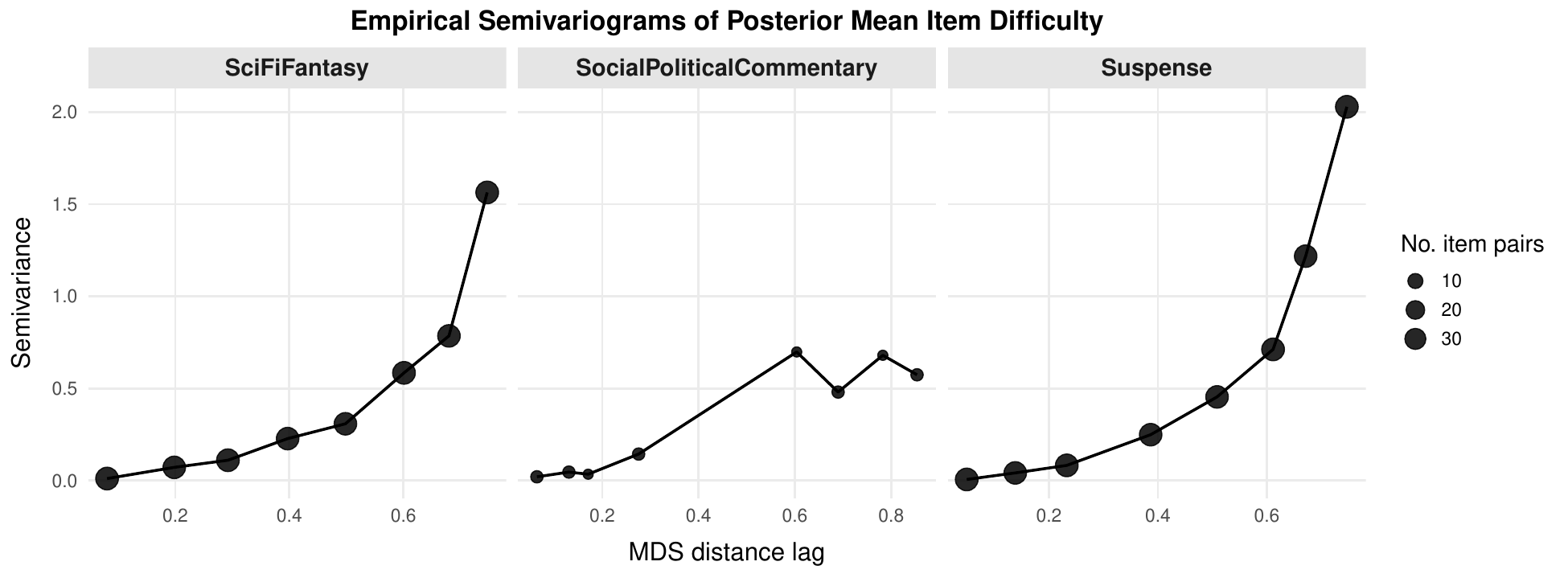}
    \caption{Empirical semivariograms of posterior mean item difficulty as a function of MDS distance, by genre. Increasing semivariance with distance indicates that items nearby in the MDS embedding tend to have more similar estimated difficulties than distant items. The pattern is strongest for \texttt{SciFiFantasy} and \texttt{Suspense}; \texttt{SocialPoliticalCommentary} contains only ten items and is interpreted cautiously.}
    \label{fig:art_variogram}
\end{figure}

\subsubsection{Genre-Level Operating Characteristics}

Moving beyond individual items, we examine how the three genres function collectively across the latent recognition-ability continuum. Write \(p_j(\theta;0)=\operatorname{logit}^{-1}(\alpha_j\theta-\beta_j)\) for the probability that item \(j\) is recognized at ability \(\theta\), with the standardized vocabulary covariate held at its calibration mean \(x=0\). For genre $g$, the average endorsement-probability curve is
\[
\bar P_g(\theta;0)
=
\frac{1}{J_g}
\sum_{j\in g}
p_j(\theta;0),
\]
where \(J_g\) is the number of items in genre \(g\). Higher curves indicate genres whose items are, on average, easier to recognize at a given ability level.

\begin{figure}[ht]
\centering
\includegraphics[width=0.8\textwidth]{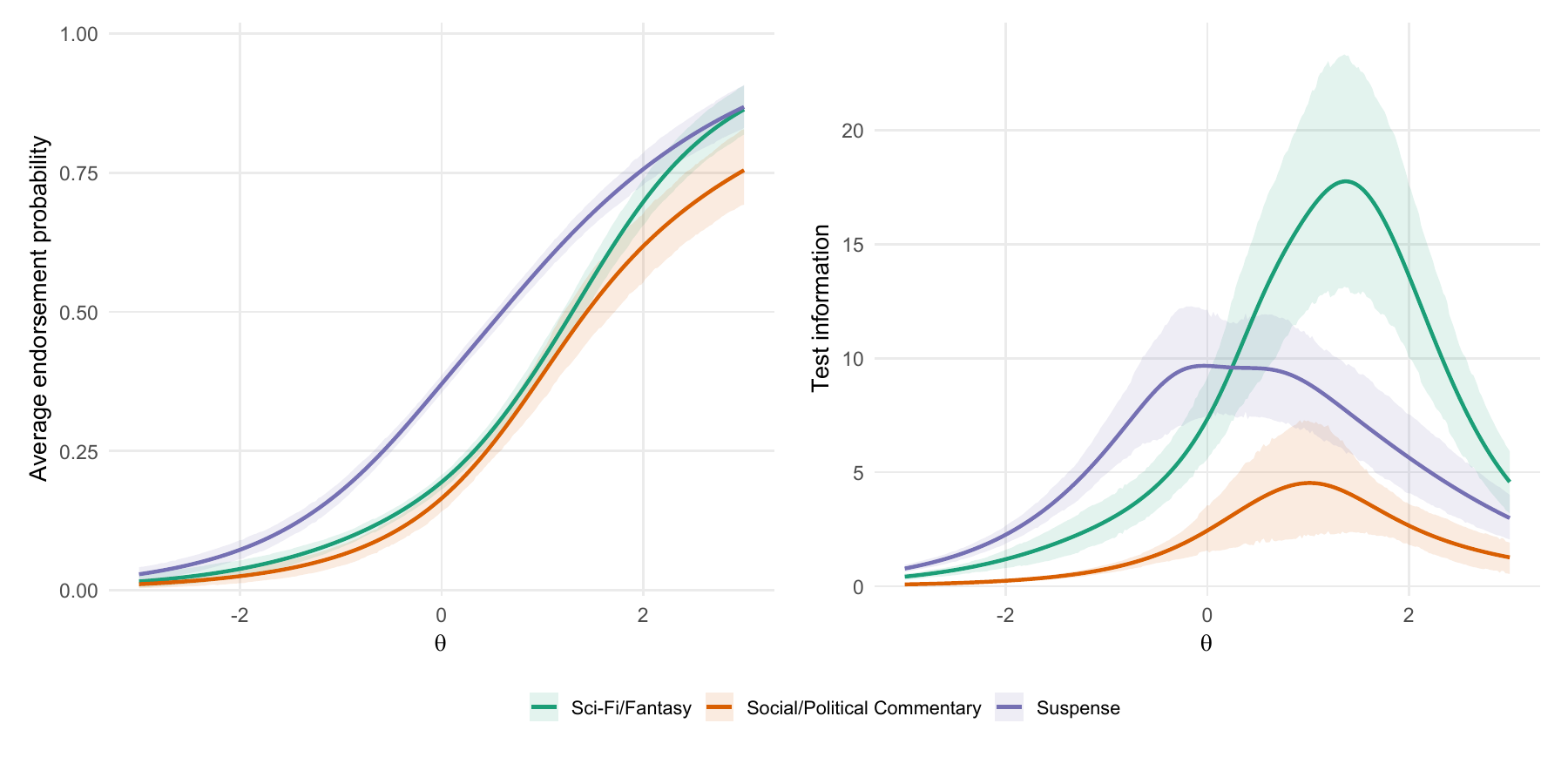}
\caption{Posterior mean genre-level average endorsement-probability curves (left) and total information curves (right) under the \texttt{SGP-IRT} model, evaluated at the mean standardized vocabulary score. Shaded regions denote $95\%$ shortest posterior intervals. Higher endorsement-probability curves indicate easier items, whereas larger information values indicate greater local measurement precision for the latent ability parameter.}
\label{fig:combined_ICC_TIC}
\end{figure}

The left panel of Figure \ref{fig:combined_ICC_TIC} shows a clear ordering of the genres. Across most of the ability continuum, \texttt{Suspense} has the highest average endorsement probability and therefore contains the easiest items overall. In contrast, \texttt{SocialPoliticalCommentary} has the lowest endorsement probability and is the most difficult genre. The \texttt{SciFiFantasy} curve generally lies between these two, although it approaches the \texttt{Suspense} curve at higher ability levels. The relatively steep increase in the \texttt{SciFiFantasy} curve suggests that recognition of these items changes particularly rapidly among participants with moderate to above-average ability.

To assess measurement precision, we also calculate the genre-level total information curve
\[
\mathcal I_g(\theta;0)
=
\sum_{j\in g}
\alpha_j^2
p_j(\theta;0)\{1-p_j(\theta;0)\}.
\]
The right panel of Figure \ref{fig:combined_ICC_TIC} demonstrates that the genres target different regions of the ability continuum. The \texttt{Suspense} genre provides its greatest information near average ability, with a peak of approximately \(9\)--\(10\) around \(\theta=0\), and retains appreciable information over a relatively broad range of lower and moderate ability values. By comparison, \texttt{SciFiFantasy} is more strongly targeted toward above-average ability: its information rises sharply after \(\theta=0\) and reaches a higher peak of approximately \(17\)--\(18\) around \(\theta=1.3\). The posterior interval surrounding this peak is comparatively wide, however, indicating greater uncertainty about its magnitude. The \texttt{SocialPoliticalCommentary} genre provides less total information, peaking at approximately \(4\)--\(5\) around \(\theta=1\).

At the respective information peaks, the conventional local information-based standard-error approximation, \(\mathcal I_g(\theta)^{-1/2}\), is approximately \(0.24\) for \texttt{SciFiFantasy}, \(0.32\) for \texttt{Suspense}, and \(0.47\) for \texttt{SocialPoliticalCommentary}. These quantities are local approximations derived from the test information and should not be interpreted as posterior standard deviations from the full \texttt{SGP-IRT} model.

Because total information is obtained by summing item information, its magnitude depends partly on the number of items in each genre. In the present data, \texttt{SciFiFantasy} and \texttt{Suspense} each contain 25 items, whereas \texttt{SocialPoliticalCommentary} contains 10. Consequently, the lower total information for \texttt{SocialPoliticalCommentary} should not be interpreted solely as lower per-item measurement quality. The results suggest that \texttt{Suspense} is most informative for participants near the center of the ability distribution, whereas \texttt{SciFiFantasy} provides greater precision for participants with above-average recognition ability.

\subsubsection{Out-of-Sample Prediction Analysis}

We evaluated predictive performance using five-fold cross-validation at the respondent level. Within each fold, the held-out respondents were excluded from both construction of the MDS item geometry and model calibration. For each held-out respondent, a prespecified subset of item responses was used as anchor information for estimating the respondent's latent ability, and predictions were evaluated on the remaining target responses. All methods used the same respondent folds, MDS coordinates, calibration samples, and standardized continuous vocabulary covariate.

We compared \texttt{SGP-IRT} with the exact published implementations of \texttt{1PLUS}, \texttt{2PLUS}, and \texttt{3PLUS}. For a fair comparison, the same standardized continuous vocabulary score was entered linearly in every model, with a corresponding regression coefficient estimated under each procedure. Predictive performance was measured using the mean log predictive score, Brier score, area under the receiver operating characteristic curve (AUC), and classification accuracy. Higher log scores, AUC values, and accuracies indicate better performance, whereas lower Brier scores are preferred. Table \ref{tab:out_of_sample_prediction} reports the mean and standard error across the five folds.

\begin{table}[htbp]
\centering
\caption{Five-fold out-of-sample predictive performance for the ART data. Entries are fold means with standard errors in parentheses. Boldface marks the metrics on which \texttt{SGP-IRT} attains the best value within a genre.}
\label{tab:out_of_sample_prediction}
\small
\setlength{\tabcolsep}{4pt}
\resizebox{\textwidth}{!}{
\begin{tabular}{llcccc}
\toprule
\textbf{Genre} & \textbf{Method}
& \textbf{Mean log score}
& \textbf{Brier score}
& \textbf{AUC}
& \textbf{Accuracy} \\ 
\midrule
\multirow{4}{*}{\texttt{SciFiFantasy}}
& \texttt{SGP-IRT} & $-0.2965\ (0.0098)$ & $0.0893\ (0.0029)$ & $0.9151\ (0.0068)$ & $0.8771\ (0.0063)$ \\
& \texttt{1PLUS} & $-0.2972\ (0.0111)$ & $0.0899\ (0.0034)$ & $0.9156\ (0.0066)$ & $0.8782\ (0.0050)$ \\
& \texttt{2PLUS} & $-0.2931\ (0.0103)$ & $0.0888\ (0.0032)$ & $0.9173\ (0.0067)$ & $0.8781\ (0.0046)$ \\
& \texttt{3PLUS} & $-0.2931\ (0.0087)$ & $0.0876\ (0.0030)$ & $0.9146\ (0.0059)$ & $0.8811\ (0.0066)$ \\
\midrule
\multirow{4}{*}{\texttt{SocialPoliticalCommentary}}
& \texttt{SGP-IRT} & $-0.3920\ (0.0208)$ & $\boldsymbol{0.1224}\ (0.0075)$ & $\boldsymbol{0.8265}\ (0.0052)$ & $0.8247\ (0.0163)$ \\
& \texttt{1PLUS} & $-0.3957\ (0.0227)$ & $0.1237\ (0.0081)$ & $0.8225\ (0.0066)$ & $0.8235\ (0.0155)$ \\
& \texttt{2PLUS} & $-0.3921\ (0.0220)$ & $0.1227\ (0.0077)$ & $0.8264\ (0.0068)$ & $0.8247\ (0.0159)$ \\
& \texttt{3PLUS} & $-0.3910\ (0.0194)$ & $0.1224\ (0.0073)$ & $0.8231\ (0.0052)$ & $0.8266\ (0.0148)$ \\
\midrule
\multirow{4}{*}{\texttt{Suspense}}
& \texttt{SGP-IRT} & $-0.4071\ (0.0038)$ & $0.1302\ (0.0024)$ & $0.8881\ (0.0049)$ & $0.8123\ (0.0034)$ \\
& \texttt{1PLUS} & $-0.4066\ (0.0047)$ & $0.1306\ (0.0026)$ & $0.8877\ (0.0053)$ & $0.8138\ (0.0051)$ \\
& \texttt{2PLUS} & $-0.4057\ (0.0044)$ & $0.1299\ (0.0026)$ & $0.8889\ (0.0054)$ & $0.8123\ (0.0042)$ \\
& \texttt{3PLUS} & $-0.4059\ (0.0052)$ & $0.1297\ (0.0027)$ & $0.8886\ (0.0058)$ & $0.8135\ (0.0039)$ \\
\bottomrule
\end{tabular}
}
\end{table}

The four methods exhibit very similar out-of-sample performance, and no method dominates uniformly across genres or scoring criteria. For \texttt{SciFiFantasy}, \texttt{2PLUS} and \texttt{3PLUS} tie on the mean log score, \texttt{2PLUS} has the highest AUC, and \texttt{3PLUS} has the lowest Brier score and the highest accuracy. For \texttt{SocialPoliticalCommentary}, \texttt{SGP-IRT} ties \texttt{3PLUS} on the Brier score and has the highest AUC, whereas \texttt{3PLUS} has the highest log score and accuracy; these differences are negligible at the displayed precision. For \texttt{Suspense}, the PLUS models have small numerical advantages on individual criteria, but the performance of all four methods remains closely comparable. Fold-paired comparisons reinforce this conclusion. Across the three genres, differences relative to \texttt{SGP-IRT} are at most \(0.0037\) in mean log score, \(0.0017\) in Brier score, \(0.0040\) in AUC, and \(0.0040\) in accuracy. Moreover, descriptive \(95\%\) intervals constructed from the five paired fold differences include zero for every method, genre, and reported metric.

The prediction analysis therefore supports a conclusion of competitive rather than uniformly better predictive performance for \texttt{SGP-IRT}. Two features of these data limit what spatial pooling can add. Each item is answered by 321 respondents, so its difficulty is well determined by its own responses, and the posterior mean under the Gaussian process stays close to the estimate obtained without pooling. Borrowing strength across items matters most when responses per item are few, which is not the case here. Prediction for a held-out respondent also runs mainly through the estimated ability, so a small change in difficulty moves the predicted probability without often moving it across the classification threshold. The same pattern appears in the simulation, where the gain in probability RMSE exceeds the gain in thresholded accuracy. The principal empirical contribution is that comparable prediction is obtained while providing a continuous, uncertainty-quantified representation of the item-difficulty surface over the estimated item geometry.

\section{Discussion}\label{sec:disc}

We have presented \texttt{SGP-IRT}, a Gaussian process extension of IRT that treats spatial dependence among item difficulties as part of the measurement model. The model addresses features that the existing spatial IRT methods leave less developed. These are the restriction to binary outcomes, the prespecified correlation structure of conditional autoregressive priors, and the limited sharing of information across the item surface. It applies whether the item coordinates are geographic or recovered from response patterns. In the simulations, it recovers item difficulty with the lowest error in every scenario we examine, with better-calibrated intervals. In the Author Recognition Test, its out-of-sample predictions are comparable to the CAR-based models, while it also yields a difficulty surface that can be read for substantive meaning.

The theory and the simulation point in the same direction. Theorem \ref{thm:sgp_irt_contraction} shows that the posterior for the difficulty surface concentrates on the truth as respondents and items grow together. The rate combines two parts. One is a smooth-surface term governed by the surface's smoothness. The other is the cost of estimating ability and discrimination jointly, which the joint-growth condition forces to vanish. This should matter more as item pools grow, a trend encouraged by computerized adaptive testing and large-scale assessment. The anisotropic kernel also lets the data set a separate length-scale per dimension, so the model can adapt to direction-specific structure without manual tuning. The simulations support this beyond smooth truth. \texttt{SGP-IRT} recovers difficulty best under rough Mat\'ern and piecewise-smooth surfaces, and also when the truth follows the neighbor-graph structure the CAR models assume.

\subsection{Implications for Psychometric Practice}

The polytomous extension carries the model beyond binary tests to any fixed number of response categories, though matching a specific ordinal model would need its own identification constraints. Latent coordinates recovered from a response-similarity matrix broaden it further, since they place spatial dependence modeling within reach of constructs that have no geographic layout but where items still share conceptual or cognitive proximity.

The Author Recognition Test illustrates this second point. The genre structure that emerges from the embedding corresponds to recognizable categories in literary knowledge, which suggests the recovered surface can carry substantive meaning rather than serving only as calibration. For test developers, item difficulty then reads as a position on a continuous surface, which may help in selecting items that cover the construct without redundancy.

The information curves in Figure \ref{fig:combined_ICC_TIC} give a concrete example for adaptive testing. \texttt{SciFiFantasy} and \texttt{Suspense} items carry the most information near $\theta=0$, while \texttt{SocialPoliticalCommentary} items offer lower but more even precision. In a computerized adaptive setting, the recovered spatial structure could inform item selection by discouraging repeated sampling from the same region of the surface, though we have not tested this directly.

\subsection{Limitations and Future Directions}

The empirical analysis reported here covers the Author Recognition Test, where the item coordinates are recovered from responses. A citizen science application with observed coordinates and a larger item pool remains for future work.

The main practical limitation is computational. Sampling cost grows with the cube of the number of locations because of the covariance operations, so for item pools in the hundreds an approximate Gaussian process, such as an inducing-point or sparse variational scheme, would likely be needed.

A natural extension adds a temporal layer. In longitudinal assessment, item dependence may shift over time, and respondent ability may change across occasions, and a spatio-temporal version of the prior could model dependence among items and within respondents together. Related directions include non-stationary kernels for surfaces whose smoothness varies across the space, nonlinear or supervised embeddings that capture structure linear MDS can miss, and hierarchical versions that allow several scales of dependence at once, as when items sit within testlets within content domains. We also fit each genre separately in the application. Fit this way, the Gaussian process estimates two length scales, a marginal scale, and a nugget from 25, 25, or 10 items, which is little information for a covariance structure. Jointly modeling all genres under one Gaussian process could let smaller genres borrow strength from larger ones, which we leave to future work. A second setting would test pooling more directly. Predicting responses to an item held out entirely requires a difficulty at a location where no responses were observed, which a continuous covariance supplies and a neighbor graph does not. That comparison speaks to item development, where new items must be placed before they are administered, and we did not carry it out here.

The recovered difficulty surface also points toward differential item functioning (DIF). Because the surface and its uncertainty describe how items behave across a continuous space, regions where the surface shifts between groups could flag candidate items for differential functioning without predefined comparison groups. We do not carry out such an analysis here, and treating it carefully would require an estimand for group-varying surfaces and a way to separate genuine differential functioning from sampling variation. We note it as a direction that the present results make available rather than as a settled use.

Finally, the model assumes that the spatial dependence it captures is informative. When some of that dependence instead reflects shared item format or surface wording, a useful extension would separate informative from nuisance dependence, perhaps by combining the Gaussian process prior with methods developed for local independence violations.

\section*{Acknowledgments}
\input{acknow}

\paragraph{Funding.} This research is not supported by any grant or funding.

\paragraph{Competing interests.} The authors declare no competing interests.

\paragraph{Data availability.} Replication code and data will be available at a public GitHub repository (Blinded for review).

\printbibliography

\appendix
\section{Proofs}

\input{new_proofs}

\numberwithin{figure}{section}
\numberwithin{table}{section}
\renewcommand{\theequation}{\thesection.\arabic{equation}}
\setcounter{figure}{0}
\setcounter{equation}{0}
\setcounter{table}{0}

\section{Additional Simulation Details}
\label{app:additional_sim}

Table \ref{tab:sim_additional_scalar_recovery} summarizes recovery of the regression coefficient \(\gamma\) and the discrimination-scale hyperparameter \(\sigma_a\). Recovery of \(\gamma\) is generally accurate, with RMSE ranging from \(0.023\) under the high-nugget anisotropic-GP truth to \(0.060\) under the i.i.d.\ truth. Empirical coverage is close to or above the nominal level under the two anisotropic-GP mechanisms and the CAR truth, ranging from \(0.940\) to \(1.000\). Coverage is lower under the rough Mat\'ern, piecewise-smooth, and i.i.d.\ mechanisms, where it equals \(0.840\), \(0.800\), and \(0.840\), respectively. Because the corresponding interval widths remain similar across all six scenarios, ranging only from \(0.170\) to \(0.184\), this pattern suggests that uncertainty quantification for \(\gamma\) is more sensitive than its point estimation to departures from the fitted spatial covariance structure.

Recovery of \(\sigma_a\) is more uniform across the six mechanisms. Its RMSE ranges narrowly from \(0.040\) to \(0.048\), while empirical coverage ranges from \(0.945\) to \(1.000\). The average interval widths lie between \(0.203\) and \(0.225\), with posterior standard deviations between \(0.053\) and \(0.058\). Thus, the marginal dispersion of the log-discrimination parameters is recovered stably even when the fitted spatial model is misspecified.

\begin{table}[h]
\centering
\caption{Recovery summaries for the regression coefficient and discrimination-scale hyperparameter under \texttt{SGP-IRT}. RMSE is computed from posterior-mean estimates across independent replications,
coverage is the empirical proportion of \(95\%\) shortest posterior
intervals containing the generating value, and interval width and posterior
SD are averaged across replications.}
\label{tab:sim_additional_scalar_recovery}
\scriptsize
\setlength{\tabcolsep}{4.5pt}
\resizebox{\textwidth}{!}{%
\begin{tabular}{llrrrrrr}
\toprule
\textbf{Scenario} & \textbf{Parameter} & \textbf{Truth} & \textbf{RMSE} & \textbf{Coverage} & \textbf{Interval width}
& \textbf{Posterior SD} \\
\midrule
\multirow{2}{*}{Anisotropic GP, low nugget}
& $\gamma$ & 0.500 &  0.029 & 0.960 & 0.170 & 0.044 \\
& $\sigma_a$ & 0.350 & 0.048 & 1.000 & 0.203 & 0.053 \\
\addlinespace
\multirow{2}{*}{Anisotropic GP, high nugget}
& $\gamma$ & 0.500 & 0.023 & 0.940 & 0.184 & 0.048 \\
& $\sigma_a$ & 0.350  & 0.040 & 0.980 & 0.212 & 0.056 \\
\addlinespace
\multirow{2}{*}{CAR truth}
& $\gamma$ & 0.500 & 0.029 & 1.000 & 0.170 & 0.044 \\
& $\sigma_a$ & 0.350 & 0.048 & 0.980 & 0.203 & 0.053 \\
\addlinespace
\multirow{2}{*}{Rough Mat\'ern}
& $\gamma$ & 0.500  & 0.043 & 0.840 & 0.179 & 0.047 \\
& $\sigma_a$ & 0.350 & 0.042 & 0.960 & 0.225 & 0.058 \\
\addlinespace
\multirow{2}{*}{Piecewise-smooth}
& $\gamma$ & 0.500 & 0.048 & 0.800 & 0.176 & 0.046 \\
& $\sigma_a$ & 0.350 & 0.046 & 1.000 & 0.222 & 0.056 \\
\addlinespace
\multirow{2}{*}{IID, no spatial signal}
& $\gamma$ & 0.500  & 0.060 & 0.840 & 0.177 & 0.046 \\
& $\sigma_a$ & 0.350 & 0.048 & 0.940 & 0.220 & 0.056 \\
\bottomrule
\end{tabular}%
}
\end{table}

\paragraph{Recovery of the GP covariance structure.}
Table \ref{tab:sim_gp_covariance_recovery} reports covariance-parameter recovery for the two scenarios generated from the anisotropic squared-exponential covariance family used by \texttt{SGP-IRT}. Under the low-nugget truth, recovery is particularly accurate for the first length scale and the derived nugget ratio, with RMSEs of \(0.020\) and \(0.007\), respectively. The second length scale has RMSE \(0.077\), while the anisotropy ratio \(\ell_2/\ell_1\) has RMSE \(0.150\). Recovery of the two scale parameters is also reasonable, with RMSE \(0.185\) for \(\sigma_f\) and \(0.069\) for \(\sigma_u\). These results indicate that, when most of the item-level variation is spatially smooth, the model can distinguish the smooth and item-specific components and recover the direction-specific correlation ranges with relatively little error.

As expected, the high-nugget setting provides less information about the spatial covariance structure. Although recovery of the marginal scales remains stable, with RMSEs of \(0.110\) for \(\sigma_f\) and \(0.067\) for \(\sigma_u\), errors increase for the direction-specific length scales, the nugget ratio, and especially the anisotropy ratio. The corresponding RMSEs are \(0.082\) for \(\ell_1\), \(0.105\) for \(\ell_2\), \(0.095\) for \(r_u\), and \(0.491\) for \(\ell_2/\ell_1\). Posterior uncertainty also increases substantially: for example, the interval width for \(\ell_2/\ell_1\) increases from \(0.415\) in the low-nugget setting to \(1.567\) in the high-nugget setting, while the widths for \(\ell_1\) and \(\ell_2\) increase from \(0.129\) and \(0.320\) to \(0.306\) and \(0.520\), respectively. This behavior is expected because a dominant item-specific component leaves less information for separately identifying the smooth spatial scale and its direction-specific ranges.

All reported covariance parameters attain empirical coverage of \(1.000\). This should be interpreted together with the interval widths -- particularly in the high-nugget setting, the uniform coverage reflects conservative posterior uncertainty rather than unusually precise identification of every covariance parameter. Nevertheless, the point-estimation results clearly distinguish the low- and high-nugget regimes. Overall, the model recovers the broad covariance decomposition, while individual length-scale and anisotropy parameters become less precisely identified when the nugget is large. This finding is consistent with the posterior dependence among the sill, nugget, and length scales discussed in Section \ref{subsec:computation} and reinforces the importance of interpreting these covariance parameters jointly rather than relying on any single posterior summary.

\begin{table}[h]
\centering
\caption{Finite-sample recovery of \texttt{SGP-IRT} covariance parameters
under correctly specified anisotropic-GP scenarios.}
\label{tab:sim_gp_covariance_recovery}
\scriptsize
\setlength{\tabcolsep}{4pt}
\resizebox{\textwidth}{!}{%
\begin{tabular}{llrrrrrr}
\toprule
\textbf{Scenario} & \textbf{Parameter} & \textbf{Truth} & \textbf{RMSE} & \textbf{Coverage} & \textbf{Interval width}
& \textbf{Posterior SD} \\
\midrule
\multirow{6}{*}{Anisotropic GP, low nugget}
& $\sigma_f$ & 1.000 & 0.185 & 1.000 & 0.814 & 0.219 \\
& $\sigma_u$ & 0.150 & 0.069 & 1.000 & 0.214 & 0.059 \\
& $\ell_1$ & 0.150 & 0.020 & 1.000 & 0.129 & 0.025 \\
& $\ell_2$ & 0.450 &  0.077 & 1.000 & 0.320 & 0.075 \\
& $r_u$ & 0.022 & 0.007 & 1.000 & 0.007 & 0.023 \\
& $\ell_2/\ell_1$ & 3.000 & 0.150 & 1.000 & 0.415 & 0.129 \\
\addlinespace
\multirow{6}{*}{Anisotropic GP, high nugget}
& $\sigma_f$ & 0.650 &  0.110 & 1.000 & 0.994 & 0.264 \\
& $\sigma_u$ & 0.900 & 0.067 & 1.000 & 0.443 & 0.116 \\
& $\ell_1$ & 0.150 & 0.082 & 1.000 & 0.306 & 0.097 \\
& $\ell_2$ & 0.450 &  0.105 & 1.000 & 0.520 & 0.162 \\
& $r_u$ & 0.657 & 0.095 & 1.000 & 0.602 & 0.169 \\
& $\ell_2/\ell_1$ & 3.000 & 0.491 & 1.000 & 1.567 & 0.492 \\
\bottomrule
\end{tabular}%
}
\end{table}

\section{MCMC Diagnostics for the ART Analysis}
\label{app:art_mcmc}

Each full-data \texttt{SGP-IRT} model was estimated using four NUTS chains with 1,000 warmup and 1,000 retained iterations per chain. Table \ref{tab:art_mcmc_diagnostics} reports diagnostics over all monitored scalar elements of \(\boldsymbol\theta\), \(\boldsymbol\alpha\), \(\boldsymbol\beta\), \(\mathbf f\), \(\mathbf u\), \(\boldsymbol\gamma\), and the GP hyperparameters. All three fits satisfy the prespecified \(\widehat R<1.01\) criterion and have minimum bulk and tail ESS values exceeding 400. No divergent transitions or maximum-treedepth exceedances were observed, and the minimum E-BFMI was comfortably above the threshold of 0.30.

\begin{table}[ht]
\centering
\caption{MCMC diagnostics for the full-data \texttt{SGP-IRT} fits. The ESS columns give the minimum over all monitored scalar parameters, while E-BFMI gives the minimum over the four chains.}
\label{tab:art_mcmc_diagnostics}
\small
\begin{tabular}{lrrrrrrr}
\toprule
Genre &
Max. $\widehat R$ &
Min. bulk ESS &
Min. tail ESS &
Divergences &
Depth hits &
Min. E-BFMI \\
\midrule
\texttt{SciFiFantasy}
& 1.0096  & 742 & 1,288 & 0 & 0 & 0.643 \\
\texttt{SocialPoliticalCommentary}
& 1.0060  & 850 & 973 & 0 & 0 & 0.581 \\
\texttt{Suspense}
& 1.0059 & 724 & 737 & 0 & 0 & 0.550 \\
\bottomrule
\end{tabular}
\end{table}

Figure \ref{fig:art_trace_scifi} displays trace plots for selected \texttt{SGP-IRT} parameters in the \texttt{SciFiFantasy} analysis, including the two length scales, anisotropy ratio, nugget ratio, spatial and item-specific standard deviations, discrimination-scale parameter, and vocabulary coefficient. We show this genre because it had the largest maximum \(\widehat R\) among the three ART fits. The four chains overlap throughout the retained sampling period and exhibit no visible chain-specific trends or persistent nonstationarity.

\begin{figure}[htbp]
\centering
\includegraphics[width=\textwidth]{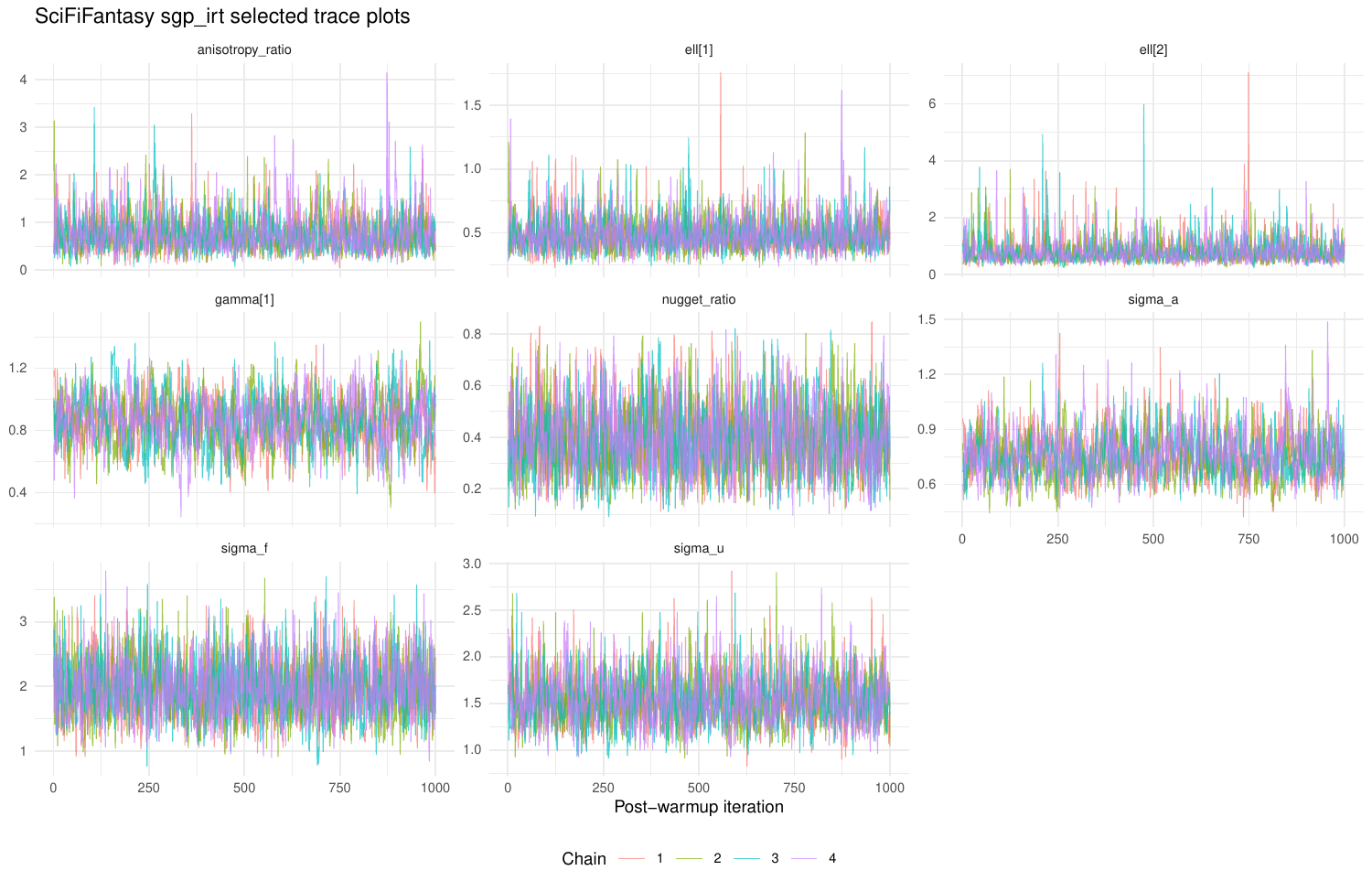}
\caption{Selected post-warmup trace plots for the
\texttt{SciFiFantasy} \texttt{SGP-IRT} fit. Colors denote the four independently initialized chains. Numerical diagnostics for all three genres are reported in Table~\ref{tab:art_mcmc_diagnostics}.}
\label{fig:art_trace_scifi}
\end{figure}

\end{document}

%% file: acknow.tex
We would like to express our sincere gratitude to Dr.\ Daniel Bolt for his insightful discussion and invaluable suggestions, which have been crucial in shaping the direction and improving the quality of this study. We thank Dr.\ Weicong Lyu for providing helpful feedback for the construction of this manuscript. We also appreciate the Center for High Throughput Computing (CHTC) in the Department of Computer Sciences at UW--Madison for their technical support in this research, \url{https://doi.org/10.21231/gnt1-hw21}. We used ChatGPT (OpenAI, 2026) solely for language editing (grammar, wording, and clarity). All outputs were reviewed and revised by the authors, who take full responsibility for the content.

%% file: new_proofs.tex
This appendix proves Theorem \ref{thm:sgp_irt_contraction} directly under the Bernoulli IRT likelihood and proves
Corollary \ref{cor:poly_contraction} for a fixed number of response categories.  

\subsection*{Notation and the two GP rate sequences}
\label{app:proof-notation}

We write
\(
  \boldsymbol\vartheta
  =
  (\boldsymbol\theta,\boldsymbol\alpha,f),
  \
  \eta_{ij}(\boldsymbol\vartheta)
  =
  \alpha_j\theta_i-f(\mathbf s_j),
\)
and let
\(
  P_{\boldsymbol\vartheta}^{(n)}
  =
  \bigotimes_{i=1}^{I}\bigotimes_{j=1}^{J}
  \operatorname{Bernoulli}
  \!\left[
    \sigma\{\eta_{ij}(\boldsymbol\vartheta)\}
  \right],\) and $n=IJ.$


The hierarchical squared-exponential GP result naturally uses two sequences:
one for prior concentration and the sieve tail, and a possibly slower one for
sieve entropy.  We choose constants $\kappa_{-}>
  {\nu(D+1)}/{(2\nu+D)},$ $\kappa_{+}>
  \kappa_{-}+{(D+1)}/{2},$ and set
\begin{align}
  \widetilde\varepsilon_n^{\,\mathrm{GP}}
  &=
  n^{-\nu/(2\nu+D)}(\log n)^{\kappa_{-}},
  \label{eq:proof-gp-small-rate}\\
  \overline\varepsilon_n^{\,\mathrm{GP}}
  &=
  n^{-\nu/(2\nu+D)}(\log n)^{\kappa_{+}},
  \label{eq:proof-gp-entropy-rate}\\
  \rho_{I,J}
  &=
  \left\{\frac{(I+J)\log n}{n}\right\}^{1/2},
  \label{eq:proof-nuisance-rate}\\
  \widetilde\varepsilon_{I,J}
  &=
  \max\left\{
    \widetilde\varepsilon_n^{\,\mathrm{GP}},
    \rho_{I,J}
  \right\},
  \label{eq:proof-denominator-rate}\\
  \varepsilon_{I,J}
  &=
  \max\left\{
    \overline\varepsilon_n^{\,\mathrm{GP}},
    \rho_{I,J}
  \right\}.
  \label{eq:proof-testing-rate}
\end{align}
Thus
\(\widetilde\varepsilon_{I,J}\le\varepsilon_{I,J}\).
The rate \(\varepsilon_{I,J}\) equals the rate
\(\varepsilon_{I,J}^{\mathrm{item}}\), up to an immaterial change in the logarithmic exponent \(\kappa\).  

\begin{remark}[Exact use of the bounded-logit assumption]
\label{rem:bounded-logit-use}
Assumption (A4) supplies measurable regularity sieves
\(\mathcal R_{I,J}\) on which all candidate logits are bounded by the same
fixed constant \(L\), together with the stated exponentially small prior
tail.  Lemma~\ref{lem:proof-bernoulli-bounds} therefore gives constants
\(0<c_{\mathrm H}<C_{\mathrm H}<\infty\), independent of \(I,J\), such that
the averaged Hellinger metric defined in
\eqref{eq:proof-average-hellinger} satisfies
\begin{equation}
  c_{\mathrm H}
  d_n^2(\boldsymbol\eta,\widetilde{\boldsymbol\eta})
  \le
  h_n^2(P_{\boldsymbol\eta},P_{\widetilde{\boldsymbol\eta}})
  \le
  C_{\mathrm H}
  d_n^2(\boldsymbol\eta,\widetilde{\boldsymbol\eta})
  \label{eq:proof-uniform-metric-equivalence}
\end{equation}
for all parameters in \(\mathcal R_{I,J}\).  Moreover, (A4) gives
\begin{equation}
  \Pi(\mathcal R_{I,J}^{\,c})
  \le
  \exp\{-C_{\mathcal R}n
  \varepsilon_{I,J}^{\,2}\}
  \label{eq:proof-regularity-tail}
\end{equation}
with \(C_{\mathcal R}\) chosen larger than the denominator constant used
below.  Under the exact ability-centering constraint, the fixed bound in
(A4) also gives
\[
  |f(\mathbf s_j)|
  =
  \left|
    \frac1I\sum_{i=1}^{I}\eta_{ij}
  \right|
  \le L,
  \qquad j=1,\ldots,J,
\]
and hence the vector
\(\{f(\mathbf s_1),\ldots,f(\mathbf s_J)\}\) has \(J\)-dimensional
covering entropy of order \(J\log(C/\varepsilon)\).  If the logit bound is
allowed to grow with \(I,J\), then Assumption~(A4) must explicitly guarantee
both the uniform constants in
\eqref{eq:proof-uniform-metric-equivalence} and the same stated entropy
order.
\end{remark}

\subsection*{Bernoulli likelihood geometry}
For \(x\in\mathbb R\), let $p_x=\sigma(x), \ A(x)=\log(1+e^x).$ For probability laws \(P_0,P\), define
\[
  K(P_0,P)
  =
  \mathbb E_0\log\frac{dP_0}{dP},
  \quad
  V(P_0,P)
  =
  \mathbb E_0\left[
    \log\frac{dP_0}{dP}-K(P_0,P)
  \right]^2.
\]
Thus \(V\) is the centered second moment of the log likelihood ratio.  

\begin{lemma}[Quadratic Bernoulli bounds]
\label{lem:proof-bernoulli-bounds}
Let \(P_x=\operatorname{Bernoulli}(p_x)\).  For all
\(x_0,x\in\mathbb R\),
\begin{align}
  K(P_{x_0},P_x)
  &\le \frac18(x-x_0)^2,
  \label{eq:proof-global-kl}\\
  V(P_{x_0},P_x)
  &\le \frac14(x-x_0)^2.
  \label{eq:proof-global-v}
\end{align}
If \(|x_0|,|x|\le L\), then there are constants
\(0<c_L<C_L<\infty\) such that
\begin{equation}
  c_L(x-x_0)^2
  \le
  h^2(P_{x_0},P_x)
  \le
  C_L(x-x_0)^2,
  \label{eq:proof-hellinger-quadratic}
\end{equation}
where
\[
  h^2(P,Q)
  =
  \int(\sqrt{dP}-\sqrt{dQ})^2.
\]
\end{lemma}

\begin{proof}
The Bernoulli model in its natural parameter is the exponential family
\[
  p_x(y)=\exp\{yx-A(x)\},
  \qquad y\in\{0,1\}.
\]
Consequently,
\[
  K(P_{x_0},P_x)
  =
  A(x)-A(x_0)-A'(x_0)(x-x_0).
\]
Because
\[
  A''(t)=\sigma(t)\{1-\sigma(t)\}\le\frac14,
\]
Taylor's theorem proves \eqref{eq:proof-global-kl}.

Moreover,
\[
  \log\frac{p_{x_0}(Y)}{p_x(Y)}
  =
  (x_0-x)Y+A(x)-A(x_0),
\]
and therefore its centered version is
\[
  (x_0-x)\{Y-p_{x_0}\}.
\]
It follows that
\[
  V(P_{x_0},P_x)
  =
  p_{x_0}(1-p_{x_0})(x-x_0)^2
  \le\frac14(x-x_0)^2,
\]
which proves \eqref{eq:proof-global-v}.

For the Hellinger claim, put \(m_L=\sigma(-L)\).  On \([-L,L]\),
\[
  m_L\le p_x\le1-m_L,
  \qquad
  m_L(1-m_L)\le\sigma'(x)\le\frac14.
\]
The mean-value theorem gives
\[
  m_L(1-m_L)|x-x_0|
  \le |p_x-p_{x_0}|
  \le\frac14|x-x_0|.
\]
Also,
\begin{align*}
  h^2(P_{x_0},P_x)
  &=
  (\sqrt{p_x}-\sqrt{p_{x_0}})^2
  +
  \{\sqrt{1-p_x}-\sqrt{1-p_{x_0}}\}^2\\
  &=
  (p_x-p_{x_0})^2
  \left\{
    \frac{1}{(\sqrt{p_x}+\sqrt{p_{x_0}})^2}
    +
    \frac{1}{
      (\sqrt{1-p_x}+\sqrt{1-p_{x_0}})^2
    }
  \right\}.
\end{align*}
The term in braces is bounded above and below by positive constants
depending only on \(m_L\).  Combining the last two displays proves
\eqref{eq:proof-hellinger-quadratic}.
\end{proof}

For two predictor arrays, define the averaged Hellinger semimetric
\begin{equation}
  h_n^2(P_{\boldsymbol\eta},P_{\widetilde{\boldsymbol\eta}})
  =
  \frac1n\sum_{i=1}^{I}\sum_{j=1}^{J}
  h^2\!\left(
    P_{\eta_{ij}},P_{\widetilde\eta_{ij}}
  \right).
  \label{eq:proof-average-hellinger}
\end{equation}
Independence and Lemma~\ref{lem:proof-bernoulli-bounds} imply
\begin{align}
  K\!\left(P_0^{(n)},P_{\boldsymbol\vartheta}^{(n)}\right)
  &\le
  \frac18n
  d_n^2(\boldsymbol\eta,\boldsymbol\eta^0),
  \label{eq:proof-joint-kl}\\
  V\!\left(P_0^{(n)},P_{\boldsymbol\vartheta}^{(n)}\right)
  &\le
  \frac14n
  d_n^2(\boldsymbol\eta,\boldsymbol\eta^0).
  \label{eq:proof-joint-v}
\end{align}
The second inequality follows because the centered log likelihood ratios are
independent and hence their variances add.

\begin{lemma}[Exponentially powerful local tests]
\label{lem:proof-local-tests}
Suppose that \(\boldsymbol\eta^0\), \(\boldsymbol\eta_1\), and the candidate
alternatives lie in \(\mathcal R_{I,J}\).  There are constants
\(c_1,c_2>0\), independent of \(I,J\), and a test
\(\phi_{\boldsymbol\eta_1}\) such that
\begin{align}
  P_0^{(n)}\phi_{\boldsymbol\eta_1}
  &\le
  \exp\{-c_1n
  d_n^2(\boldsymbol\eta_1,\boldsymbol\eta^0)\},
  \label{eq:proof-local-type-one}\\
  \sup_{
    d_n(\boldsymbol\eta,\boldsymbol\eta_1)
    \le
    c_2d_n(\boldsymbol\eta_1,\boldsymbol\eta^0)
  }
  P_{\boldsymbol\eta}^{(n)}
  (1-\phi_{\boldsymbol\eta_1})
  &\le
  \exp\{-c_1n
  d_n^2(\boldsymbol\eta_1,\boldsymbol\eta^0)\}.
  \label{eq:proof-local-type-two}
\end{align}
\end{lemma}

\begin{proof}
For product measures, the likelihood-ratio test between
\(P_0^{(n)}\) and \(P_{\boldsymbol\eta_1}^{(n)}\) has the sum of its two
errors bounded by their Hellinger affinity.  Since
\[
  \int\sqrt{dP_0^{(n)}dP_{\boldsymbol\eta_1}^{(n)}}
  =
  \prod_{i,j}
  \left\{1-\frac12
  h^2(P_{\eta_{ij}^0},P_{\eta_{1,ij}})\right\},
\]
the inequality \(1-u\le e^{-u}\) and
\eqref{eq:proof-uniform-metric-equivalence} yield an upper bound of the
form
\(
  \exp\{-c n
  d_n^2(\boldsymbol\eta_1,\boldsymbol\eta^0)\}.
\)
The standard robustification of the likelihood-ratio test makes the
Type-II bound uniform over a Hellinger ball whose radius is a sufficiently
small fixed fraction of the separation from the null
\parencite[Lemma 2]{Ghosal_2007}.
The uniform equivalence
\eqref{eq:proof-uniform-metric-equivalence} converts that Hellinger ball to the \(d_n\)-ball in \eqref{eq:proof-local-type-two}.
\end{proof}

\begin{lemma}[Tests for an entropy-controlled shell]
\label{lem:proof-shell-tests}
Let \(\mathcal P_n\subseteq\mathcal R_{I,J}\), let \(r>0\), and suppose that
\[
  \log N(\delta r,\mathcal P_n,d_n)\le H_n(r)
\]
for a sufficiently small fixed \(\delta>0\).  There is a test \(\phi_r\)
such that
\begin{align}
  P_0^{(n)}\phi_r
  &\le
  \exp\{H_n(r)-c_3nr^2\},
  \label{eq:proof-shell-type-one}\\
  \sup_{
    \boldsymbol\eta\in\mathcal P_n:
    r<d_n(\boldsymbol\eta,\boldsymbol\eta^0)\le2r
  }
  P_{\boldsymbol\eta}^{(n)}(1-\phi_r)
  &\le
  \exp(-c_3nr^2).
  \label{eq:proof-shell-type-two}
\end{align}
\end{lemma}

\begin{proof}
Cover the shell by a minimal \(\delta r\)-net and attach the test from
Lemma~\ref{lem:proof-local-tests} to every center.  When \(\delta\) is small
enough, every point in a net ball satisfies
\eqref{eq:proof-local-type-two} for its center.  Take \(\phi_r\) to be the
maximum of the center tests.  The union bound gives
\eqref{eq:proof-shell-type-one}, and the covering property gives
\eqref{eq:proof-shell-type-two}.
\end{proof}

\paragraph{Squared-exponential RKHS.}
For a fixed inverse-length-scale vector
\(\mathbf a=(a_1,\ldots,a_D)\), let
\(
  W^{\mathbf a}(\mathbf s)=W(\mathbf A\mathbf s),
  \
  \mathbf A=\operatorname{diag}(a_1,\ldots,a_D),
\)
and denote its RKHS on \([0,1]^D\) by
\(\mathbb H^{\mathbf a}\), with unit ball
\(\mathbb H_1^{\mathbf a}\).  Let
\(\mathbb B_\infty\) be the unit ball of
\(C([0,1]^D)\).

Note that, a valid Gaussian sieve is an inflated set
\(
  M_n\mathbb H_1^{\mathbf a}
  +\delta_n\mathbb B_\infty,
\)
not the bare RKHS ball
\(\{f:\|f\|_{\mathbb H^{\mathbf a}}\le M_n\}\).

\begin{lemma}[Entropy of a squared-exponential RKHS ball]
\label{lem:proof-rkhs-entropy}
There is a constant \(C\), depending only on \(D\), such that, for
\(a_d\ge1\), \(M\ge1\), and \(0<\delta<1/2\),
\begin{equation}
  \log N\!\left(
    \delta,
    M\mathbb H_1^{\mathbf a},
    \|\cdot\|_\infty
  \right)
  \le
  C\left(\prod_{d=1}^{D}a_d\right)
  \left[
    \log\left\{
      \frac{CM\prod_{d=1}^{D}(1+a_d)}{\delta}
    \right\}
  \right]^{D+1}.
  \label{eq:proof-rkhs-entropy}
\end{equation}
\end{lemma}

\begin{proof}
The spectral measure of the squared-exponential process is Gaussian.  Every
\(h\in\mathbb H_1^{\mathbf a}\) has the representation
\[
  h(\mathbf t)
  =
  \operatorname{Re}
  \int e^{i\langle\boldsymbol\lambda,\mathbf t\rangle}
  \psi(\boldsymbol\lambda)
  \,d\mu_{\mathbf a}(\boldsymbol\lambda),
  \quad
  \|\psi\|_{L^2(\mu_{\mathbf a})}\le1.
\]
For a multi-index \(\mathbf k\), differentiation under the integral and
Cauchy--Schwarz give
\[
  \|D^{\mathbf k}h\|_\infty
  \le
  \left\{
    \int
    \prod_{d=1}^{D}|\lambda_d|^{2k_d}
    \,d\mu_{\mathbf a}(\boldsymbol\lambda)
  \right\}^{1/2}
  \le
  C^{|\mathbf k|}
  \prod_{d=1}^{D}a_d^{k_d}
  \sqrt{\mathbf k!}.
\]
Hence the RKHS unit ball is a uniformly bounded class of real-analytic
functions, with analytic radius of order \(a_d^{-1}\) in coordinate \(d\).
Partition \([0,1]^D\) into at most
\(C\prod_d a_d\) rectangles whose \(d\)th side has length of order
\(a_d^{-1}\).  On each rectangle, truncate the multivariate Taylor series at
total degree
\[
  K
  \asymp
  \log\left\{
    \frac{CM\prod_d(1+a_d)}{\delta}
  \right\}.
\]
The derivative bound and Stirling's formula make the Taylor remainder at
most \(\delta/(2M)\).  There are \(O(K^D)\) coefficients of total degree at
most \(K\).  Quantizing each coefficient finely enough to contribute at
most \(\delta/(2M)\) to the local sup-norm error requires at most
\(\exp\{CK^D\log(CM/\delta)\}\) choices per rectangle.  Multiplication over
the \(O(\prod_d a_d)\) rectangles proves
\eqref{eq:proof-rkhs-entropy}.
\end{proof}

\begin{lemma}[RKHS approximation of the H\"older truth]
\label{lem:proof-rkhs-approximation}
Let \(f_0\in\mathcal C^\nu([0,1]^D;B)\).  For the isotropic subfamily
\(\mathbf a=(a,\ldots,a)\) and every sufficiently large \(a\), there is an
\(h_a\in\mathbb H^{\mathbf a}\) such that
\begin{equation}
  \|h_a-f_0\|_\infty\le Ca^{-\nu},
  \quad
  \|h_a\|_{\mathbb H^{\mathbf a}}^2\le Ca^D.
  \label{eq:proof-rkhs-approximation}
\end{equation}
\end{lemma}

\begin{proof}
Extend \(f_0\) to a compactly supported
\(\mathcal C^\nu\)-function on \(\mathbb R^D\).  Let \(\zeta\) be a smooth
kernel whose Fourier transform is compactly supported, which integrates to
one, and whose moments through order \(\lfloor\nu\rfloor\) vanish apart
from the zeroth moment.  Define
\[
  h_a(\mathbf t)
  =
  a^D\int
  \zeta\{a(\mathbf t-\mathbf u)\}f_0(\mathbf u)
  \,d\mathbf u.
\]
Taylor expansion of \(f_0(\mathbf u)\) about \(\mathbf t\), cancellation of
the kernel moments, and the H\"older remainder yield
\(\|h_a-f_0\|_\infty\le Ca^{-\nu}\).

The Fourier transform of \(h_a\) is supported in a box of radius \(O(a)\).
On this box, the spectral density
\(\rho_{\mathbf a}\) of the rescaled squared-exponential process is bounded
below by \(ca^{-D}\).  The spectral characterization of the RKHS norm and
Plancherel's identity therefore give
\[
  \|h_a\|_{\mathbb H^{\mathbf a}}^2
  =
  \int
  \frac{|\widehat h_a(\boldsymbol\lambda)|^2}
       {\rho_{\mathbf a}(\boldsymbol\lambda)}
  \,d\boldsymbol\lambda
  \le
  Ca^D\int|\widehat h_a(\boldsymbol\lambda)|^2
  \,d\boldsymbol\lambda
  \le Ca^D.
\]
\end{proof}

For a centered Gaussian random element \(Z\) in a separable Banach space,
with RKHS \(\mathbb H_Z\), define its concentration function about \(f_0\)
by
\begin{equation}
  \phi_{f_0}^{Z}(\delta)
  =
  \inf_{
    h\in\mathbb H_Z:
    \|h-f_0\|_\infty<\delta
  }
  \|h\|_{\mathbb H_Z}^2
  -
  \log\mathbb P(\|Z\|_\infty<\delta).
  \label{eq:proof-concentration-function}
\end{equation}

\begin{lemma}[Conditional GP concentration]
\label{lem:proof-conditional-concentration}
For \(\mathbf a=(a,\ldots,a)\), \(a\ge1\), and
\(a^{-\nu}\lesssim\delta<1/2\),
\begin{equation}
  \phi_{f_0}^{W^{\mathbf a}}(2\delta)
  \le
  Ca^D
  \left\{
    \log\left(\frac{Ca}{\delta}\right)
  \right\}^{D+1}.
  \label{eq:proof-conditional-concentration}
\end{equation}
\end{lemma}

\begin{proof}
Lemma~\ref{lem:proof-rkhs-approximation} supplies
\(h_a\in\mathbb H^{\mathbf a}\) with
\(\|h_a-f_0\|_\infty\le\delta\) and
\(\|h_a\|_{\mathbb H^{\mathbf a}}^2\le Ca^D\).
The centered small-deviation bound obtained from
Lemma~\ref{lem:proof-rkhs-entropy} by the standard Gaussian
entropy--small-ball inequality is
\[
  -\log\mathbb P(
    \|W^{\mathbf a}\|_\infty\le\delta
  )
  \le
  Ca^D
  \left\{
    \log\left(\frac{Ca}{\delta}\right)
  \right\}^{D+1}.
\]
Substitution of these two bounds into
\eqref{eq:proof-concentration-function} proves the claim.
\end{proof}

\begin{lemma}[Hierarchical GP prior mass and sieve]
\label{lem:proof-hierarchical-gp}
Under Assumptions~(A2) and (A5), for every sufficiently large fixed
\(C_0>0\), there are measurable sets
\(\mathcal F_n\subset C([0,1]^D)\) such that
\begin{align}
  \Pi_f\!\left(
    \|f-f_0\|_\infty
    \le C\widetilde\varepsilon_n^{\,\mathrm{GP}}
  \right)
  &\ge
  \exp\left\{
    -C_1n
    (\widetilde\varepsilon_n^{\,\mathrm{GP}})^2
  \right\},
  \label{eq:proof-gp-prior-mass}\\
  \log N\!\left(
    C\overline\varepsilon_n^{\,\mathrm{GP}},
    \mathcal F_n,
    \|\cdot\|_\infty
  \right)
  &\le
  C_2n
  (\overline\varepsilon_n^{\,\mathrm{GP}})^2,
  \label{eq:proof-gp-sieve-entropy}\\
  \Pi_f(\mathcal F_n^c)
  &\le
  \exp\left\{
    -C_0n
    (\widetilde\varepsilon_n^{\,\mathrm{GP}})^2
  \right\}.
  \label{eq:proof-gp-sieve-tail}
\end{align}
\end{lemma}

\begin{proof}
Let \(\Pi_f^{\mathrm{raw}}\) denote the untruncated hierarchical GP law and
let \(E_f=\{\|f\|_\infty\le B_f\}\).  Assumption (A5) uses
\(\Pi_f=\Pi_f^{\mathrm{raw}}(\,\cdot\mid E_f)\), where
\(p_f=\Pi_f^{\mathrm{raw}}(E_f)>0\).  Because
\(\|f_0\|_\infty\le B<B_f\), the sup-norm ball in
\eqref{eq:proof-gp-prior-mass} is contained in \(E_f\) for all sufficiently
large \(n\).  Consequently, conditioning divides its raw prior probability
by \(p_f\), intersecting a raw sieve with \(E_f\) cannot increase its
entropy, and its conditional complement probability is at most
\(p_f^{-1}\) times the corresponding raw complement probability.  It is
therefore enough to establish the three bounds for
\(\Pi_f^{\mathrm{raw}}\), since the fixed factor \(p_f^{-1}\) is absorbed
by the constants below.

We give the raw-prior construction because it is the point at which the
RKHS results enter.  For notational simplicity the displays below take
\(\sigma_f=1\).  A fixed positive amplitude is handled by rescaling
\(f_0\) and the neighborhood radius by constants; for a random amplitude,
we restrict \(\sigma_f\) to the fixed neighborhood of one supplied by
(A5).  Under (A5),
\(A_d^{D}\) has a Gamma density.  Consequently, for large \(a,r\),
\begin{align}
  -\log\Pi\{A_d\in[a,2a],\ d=1,\ldots,D\}
  &\le Ca^D,
  \label{eq:proof-bandwidth-local-mass}\\
  \Pi\{\max_d A_d>r\}
  &\le Cr^{C'}e^{-c r^D}.
  \label{eq:proof-bandwidth-tail}
\end{align}
The same orders hold if the fixed interval \([a,2a]\) is replaced by a
smaller multiplicative neighborhood of \(a\).

Let
\(\delta_n=\widetilde\varepsilon_n^{\,\mathrm{GP}}\) and take
\(a_n\asymp\delta_n^{-1/\nu}\).  Uniformly for
\(A_d\in[a_n,2a_n]\), the approximation and centered small-ball bounds in
Lemmas~\ref{lem:proof-rkhs-approximation} and
\ref{lem:proof-conditional-concentration}, together with the
Cameron--Martin inequality, give
\[
  -\log\mathbb P\{
    \|W^{\mathbf A}-f_0\|_\infty
    \le C\delta_n
    \mid\mathbf A
  \}
  \le
  C\delta_n^{-D/\nu}
  \{\log(1/\delta_n)\}^{D+1}.
\]
The choice of \(\kappa_{-}\) such that ${\nu(D+1)}/{(2\nu+D)}$ ensures that
\begin{equation}
  \delta_n^{-D/\nu}
  \{\log(1/\delta_n)\}^{D+1}
  \lesssim
  n\delta_n^2.
  \label{eq:proof-gp-balance}
\end{equation}
After cancelling
\(n^{D/(2\nu+D)}\) from the two sides,
\eqref{eq:proof-gp-balance} is implied by
\(\kappa_{-}(2+D/\nu)>D+1\).
Integrating over the local bandwidth rectangle and using
\eqref{eq:proof-bandwidth-local-mass} proves
\eqref{eq:proof-gp-prior-mass}.  The amplitude restrictions just described
change only the constants.

For the sieve, choose \(r_n\) with $r_n^D=K n\delta_n^2$ for a sufficiently large \(K\), truncate the inverse length scales at \(\max_d A_d\le r_n\), and partition the retained bandwidth rectangle into
a finite multiplicative grid.  For each grid representative
\(\mathbf a\), use the inflated RKHS set
\(
  M_n\mathbb H_1^{\mathbf a}
  +\delta_n\mathbb B_\infty,
  \
  M_n^2=K'n\delta_n^2.
\)
Borell's Gaussian isoperimetric inequality gives a conditional probability
outside this set bounded by \(\exp(-cM_n^2)\), after \(K'\) is chosen large
enough to dominate the centered small-ball quantile.  The bandwidth tail
\eqref{eq:proof-bandwidth-tail}, the amplitude-tail condition in (A5), and
a union bound over the multiplicative grid prove
\eqref{eq:proof-gp-sieve-tail}; \(K,K'\) can be enlarged to obtain any fixed
\(C_0\).

Finally,
\[
  N\!\left(
    2\delta_n,
    M_n\mathbb H_1^{\mathbf a}
    +\delta_n\mathbb B_\infty,
    \|\cdot\|_\infty
  \right)
  \le
  N\!\left(
    \delta_n,
    M_n\mathbb H_1^{\mathbf a},
    \|\cdot\|_\infty
  \right).
\]
Applying Lemma \ref{lem:proof-rkhs-entropy} on every bandwidth cell and
adding the logarithm of the number of grid cells yields the bound
\[
  \log N\!\left(
    C\overline\varepsilon_n^{\,\mathrm{GP}},
    \mathcal F_n,
    \|\cdot\|_\infty
  \right)
  \le
  C n\delta_n^2(\log n)^{D+1}.
\]
Because
\(\kappa_{+}>\kappa_{-}+(D+1)/2\),
\[
  \delta_n^2(\log n)^{D+1}
  \lesssim
  (\overline\varepsilon_n^{\,\mathrm{GP}})^2,
\]
which proves \eqref{eq:proof-gp-sieve-entropy}.  This is the familiar
two-sequence construction for randomly rescaled squared-exponential
processes \parencite{vaartzanten2009,BhattacharyaGPAoS}.
\end{proof}

\subsection*{Nuisance parameters and the predictor map}

\begin{lemma}[Nuisance prior mass and entropy]
\label{lem:proof-nuisance}
Under Assumptions (A3) and (A5), for sufficiently small fixed \(c>0\),
\begin{align}
  &\Pi_\theta\!\left(
    \|\boldsymbol\theta-\boldsymbol\theta^0\|_I
    \le c\widetilde\varepsilon_{I,J}
  \right)
  \Pi_\alpha\!\left(
    \|\boldsymbol\alpha-\boldsymbol\alpha^0\|_J
    \le c\widetilde\varepsilon_{I,J}
  \right)
  \notag\\
  &\hspace{35mm}\ge
  \exp\left[
    -C(I+J)
    \log\left\{
      \frac{C}{\widetilde\varepsilon_{I,J}}
    \right\}
  \right],
  \label{eq:proof-nuisance-mass}
\end{align}
and the nuisance parameter space has an
\(\varepsilon\)-cover, in the product empirical norm, with
\begin{equation}
  \log N(\varepsilon,\mathcal P_{I,J}^{\mathrm{nuis}})
  \le
  C(I+J)\log(C/\varepsilon).
  \label{eq:proof-nuisance-entropy}
\end{equation}
\end{lemma}

\begin{proof}
The identified ability space is the intersection of the centered hyperplane
with the empirical unit sphere; its dimension is \(I-2\).  Its covering
number in \(\|\cdot\|_I\) is at most \((C/\varepsilon)^{I-2}\).
The discrimination space is contained in the \(J\)-dimensional compact box
\([\underline a,\overline a]^J\), whose covering number in
\(\|\cdot\|_J\) is at most \((C/\varepsilon)^J\).  Multiplication of the
two covering numbers gives \eqref{eq:proof-nuisance-entropy}.

Assumption (A5) places the priors for \(\boldsymbol\theta\) and
\(\boldsymbol\alpha\) on the identified spaces in (A3), with densities bounded
below, relative to the corresponding surface and Lebesgue measures, in
neighborhoods of their true values.  The volumes of empirical balls of
radius \(\varepsilon\) in dimensions \(I-2\) and \(J\) are bounded below,
up to exponentially dimension-dependent constants, by
\((c\varepsilon)^{I-2}\) and \((c\varepsilon)^J\).  Taking logarithms gives
\eqref{eq:proof-nuisance-mass}.
\end{proof}

\begin{lemma}[Lipschitz map from parameters to predictors]
\label{lem:proof-parameter-lipschitz}
There is a constant \(C\), depending only on
\(\overline a\), such that
\begin{align}
  d_n^2(\boldsymbol\eta,\boldsymbol\eta^0)
  \le
  C\bigl\{
    &\|\boldsymbol\theta-\boldsymbol\theta^0\|_I^2
    +
    \|\boldsymbol\alpha-\boldsymbol\alpha^0\|_J^2
    +
    \|f-f_0\|_J^2
  \bigr\}.
  \label{eq:proof-param-to-predictor}
\end{align}
For two candidate triples satisfying (A3),
\begin{align}
  d_n^2(\boldsymbol\eta,\widetilde{\boldsymbol\eta})
  \le
  C\bigl\{
    &\|\boldsymbol\theta-\widetilde{\boldsymbol\theta}\|_I^2
    +
    \|\boldsymbol\alpha-\widetilde{\boldsymbol\alpha}\|_J^2
    +
    \|f-\widetilde f\|_J^2
  \bigr\}.
  \label{eq:proof-param-to-predictor-cover}
\end{align}
\end{lemma}

\begin{proof}
Add and subtract \(\alpha_j^0\theta_i\):
\[
  \alpha_j\theta_i-\alpha_j^0\theta_i^0
  =
  (\alpha_j-\alpha_j^0)\theta_i
  +
  \alpha_j^0(\theta_i-\theta_i^0).
\]
Using \((x+y+z)^2\le3(x^2+y^2+z^2)\), averaging over
\((i,j)\), and applying the exact scaling constraint
\(\|\boldsymbol\theta\|_I=1\), we obtain
\[
  \frac1{IJ}\sum_{i,j}
  \{(\alpha_j-\alpha_j^0)\theta_i\}^2
  =
  \|\boldsymbol\alpha-\boldsymbol\alpha^0\|_J^2
\]
and
\[
  \frac1{IJ}\sum_{i,j}
  \{\alpha_j^0(\theta_i-\theta_i^0)\}^2
  \le
  \overline a^2
  \|\boldsymbol\theta-\boldsymbol\theta^0\|_I^2.
\]
This proves \eqref{eq:proof-param-to-predictor}.  Replacing the true triple
by a second candidate triple and adding and subtracting
\(\widetilde\alpha_j\theta_i\) proves
\eqref{eq:proof-param-to-predictor-cover}.
\end{proof}

\subsection*{Joint prior mass and a joint sieve}
\label{app:joint-prior-sieve}

For \(\varepsilon>0\), define the KL neighborhood
\begin{equation}
  \mathcal K_n(\varepsilon)
  =
  \left\{
    \boldsymbol\vartheta:
    K(P_0^{(n)},P_{\boldsymbol\vartheta}^{(n)})
    \le n\varepsilon^2,\quad
    V(P_0^{(n)},P_{\boldsymbol\vartheta}^{(n)})
    \le n\varepsilon^2
  \right\}.
  \label{eq:proof-kl-neighborhood}
\end{equation}

\begin{lemma}[Joint KL prior mass]
\label{lem:proof-joint-prior-mass}
There are constants \(C,C'>0\) such that
\begin{equation}
  \Pi\left\{
    \mathcal K_n(C\widetilde\varepsilon_{I,J})
  \right\}
  \ge
  \exp\left\{
    -C'n\widetilde\varepsilon_{I,J}^{\,2}
  \right\}.
  \label{eq:proof-joint-prior-mass}
\end{equation}
\end{lemma}

\begin{proof}
By Lemma~\ref{lem:proof-parameter-lipschitz}, the product neighborhood
\[
  \|\boldsymbol\theta-\boldsymbol\theta^0\|_I
  \le c\widetilde\varepsilon_{I,J},
  \quad
  \|\boldsymbol\alpha-\boldsymbol\alpha^0\|_J
  \le c\widetilde\varepsilon_{I,J},
  \quad
  \|f-f_0\|_\infty
  \le c\widetilde\varepsilon_{I,J}
\]
is contained in a
\(d_n\)-ball of radius \(C\widetilde\varepsilon_{I,J}\).
The global KL and variance bounds
\eqref{eq:proof-joint-kl}--\eqref{eq:proof-joint-v} then place this
neighborhood inside
\(\mathcal K_n(C\widetilde\varepsilon_{I,J})\).

By independence of the three prior blocks,
Lemma~\ref{lem:proof-nuisance}, and
Lemma~\ref{lem:proof-hierarchical-gp}, the negative logarithm of the prior
mass of the product neighborhood is at most
\[
  C(I+J)
  \log\left(
    \frac{C}{\widetilde\varepsilon_{I,J}}
  \right)
  +
  Cn(\widetilde\varepsilon_n^{\,\mathrm{GP}})^2.
\]
Since
\(\widetilde\varepsilon_{I,J}\ge\rho_{I,J}\),
\[
  (I+J)
  \log\left(
    \frac{C}{\widetilde\varepsilon_{I,J}}
  \right)
  \lesssim
  (I+J)\log n
  \le
  n\widetilde\varepsilon_{I,J}^{\,2}.
\]
Also
\(\widetilde\varepsilon_{I,J}\ge
\widetilde\varepsilon_n^{\,\mathrm{GP}}\).
These two observations prove
\eqref{eq:proof-joint-prior-mass}.
\end{proof}

\begin{lemma}[Joint sieve entropy and tail]
\label{lem:proof-joint-sieve}
For a fixed \(C_0>C_D\), as supplied by the tail clause in
Remark \ref{rem:bounded-logit-use}, there is a joint sieve
\(\mathcal P_n\subseteq\mathcal R_{I,J}\) such that
\begin{align}
  \log N(
    c\varepsilon_{I,J},
    \mathcal P_n,
    d_n
  )
  &\le
  Cn\varepsilon_{I,J}^{\,2},
  \label{eq:proof-joint-entropy}\\
  \Pi(\mathcal P_n^c)
  &\le
  \exp\{
    -C_0n\widetilde\varepsilon_{I,J}^{\,2}
  \}.
  \label{eq:proof-joint-tail}
\end{align}
\end{lemma}

\begin{proof}
Take \(\mathcal P_n=\mathcal R_{I,J}\).  On this sieve, cover the identified
ability and discrimination spaces as in
Lemma~\ref{lem:proof-nuisance}.  By
Remark~\ref{rem:bounded-logit-use}, the evaluation vector
\[
  \mathbf f_J
  =
  \{f(\mathbf s_1),\ldots,f(\mathbf s_J)\}
\]
lies in a fixed \(J\)-dimensional box and therefore has an
\(O(\varepsilon_{I,J})\)-cover in \(\|\cdot\|_J\) of logarithmic size at
most \(CJ\log(C/\varepsilon_{I,J})\).  The semimetric \(d_n\) depends on
\(f\) only through this evaluation vector.  By
Lemma~\ref{lem:proof-parameter-lipschitz}, the product of these three covers
is a \(c\varepsilon_{I,J}\)-cover in \(d_n\).  Thus
\[
  \log N(
    c\varepsilon_{I,J},
    \mathcal P_n,
    d_n
  )
  \le
  C(I+J)\log(C/\varepsilon_{I,J})
  \le
  Cn\rho_{I,J}^2
  \le
  Cn\varepsilon_{I,J}^{\,2},
\]
which proves \eqref{eq:proof-joint-entropy}.  The tail bound
\eqref{eq:proof-joint-tail} follows directly from
\eqref{eq:proof-regularity-tail}, because
\(\varepsilon_{I,J}\ge\widetilde\varepsilon_{I,J}\), after increasing
\(C_{\mathcal R}\) if necessary.
\end{proof}




\subsection*{Posterior contraction of the linear predictor}
\label{app:predictor-contraction}

\begin{lemma}[Marginal-likelihood denominator]
\label{lem:proof-denominator}
Let
\[
  D_n
  =
  \int
  \frac{dP_{\boldsymbol\vartheta}^{(n)}}
       {dP_0^{(n)}}(\mathbf Y)
  \,d\Pi(\boldsymbol\vartheta).
\]
For a sufficiently large constant \(C_D\),
\begin{equation}
  P_0^{(n)}
  \left\{
    D_n<
    \exp(
      -C_Dn\widetilde\varepsilon_{I,J}^{\,2}
    )
  \right\}
  \longrightarrow0.
  \label{eq:proof-denominator-bound}
\end{equation}
\end{lemma}

\begin{proof}
Let
\[
  B_n
  =
  \mathcal K_n(C\widetilde\varepsilon_{I,J})
\]
and let \(Q_n\) be the prior restricted to \(B_n\) and normalized to have
total mass one.  Jensen's inequality gives
\[
  D_n
  \ge
  \Pi(B_n)\exp(Z_n),
  \quad
  Z_n
  =
  \int_{B_n}
  \log
  \frac{dP_{\boldsymbol\vartheta}^{(n)}}
       {dP_0^{(n)}}(\mathbf Y)
  \,dQ_n(\boldsymbol\vartheta).
\]
By the definition of \(B_n\),
\[
  \mathbb E_0Z_n
  =
  -
  \int_{B_n}
  K(P_0^{(n)},P_{\boldsymbol\vartheta}^{(n)})
  \,dQ_n(\boldsymbol\vartheta)
  \ge
  -Cn\widetilde\varepsilon_{I,J}^{\,2}.
\]
Jensen's inequality in \(L^2(P_0^{(n)})\) and the definition of \(V\)
give
\begin{align*}
  \mathbb E_0
  (Z_n-\mathbb E_0Z_n)^2
  &\le
  \int_{B_n}
  V(P_0^{(n)},P_{\boldsymbol\vartheta}^{(n)})
  \,dQ_n(\boldsymbol\vartheta)\\
  &\le
  Cn\widetilde\varepsilon_{I,J}^{\,2}.
\end{align*}
For every fixed \(b>0\), Chebyshev's inequality therefore yields
\[
  P_0^{(n)}
  \left\{
    Z_n<-(C+b)n\widetilde\varepsilon_{I,J}^{\,2}
  \right\}
  \le
  \frac{C}{
    b^2n\widetilde\varepsilon_{I,J}^{\,2}
  }
  \longrightarrow0.
\]
On the complementary event,
Lemma~\ref{lem:proof-joint-prior-mass} implies
\[
  D_n
  \ge
  \exp\{
    -(C+b)n\widetilde\varepsilon_{I,J}^{\,2}
  \}
  \Pi(B_n)
  \ge
  \exp\{
    -C_Dn\widetilde\varepsilon_{I,J}^{\,2}
  \}
\]
for a sufficiently large \(C_D\).
\end{proof}

\begin{lemma}[Prediction-metric contraction]
\label{lem:proof-predictor-contraction}
For a sufficiently large constant \(M\),
\begin{equation}
  \Pi\!\left\{
    d_n(\boldsymbol\eta,\boldsymbol\eta^0)
    >
    M\varepsilon_{I,J}
    \,\middle|\,
    \mathbf Y
  \right\}
  \longrightarrow0
  \label{eq:proof-predictor-contraction}
\end{equation}
in \(P_0^{(n)}\)-probability.
\end{lemma}

\begin{proof}
For \(k=0,1,\ldots\), define the dyadic shells
\[
  S_{n,k}
  =
  \left\{
    \boldsymbol\vartheta\in\mathcal P_n:
    2^kM\varepsilon_{I,J}
    <
    d_n(\boldsymbol\eta,\boldsymbol\eta^0)
    \le
    2^{k+1}M\varepsilon_{I,J}
  \right\}.
\]
Because covering numbers decrease as the covering radius increases,
\eqref{eq:proof-joint-entropy} and
Lemma~\ref{lem:proof-shell-tests} provide tests \(\phi_{n,k}\) satisfying
\begin{align}
  P_0^{(n)}\phi_{n,k}
  &\le
  \exp\{
    Cn\varepsilon_{I,J}^{\,2}
    -
    c4^kM^2n\varepsilon_{I,J}^{\,2}
  \},
  \label{eq:proof-global-type-one}\\
  \sup_{\boldsymbol\vartheta\in S_{n,k}}
  P_{\boldsymbol\vartheta}^{(n)}
  (1-\phi_{n,k})
  &\le
  \exp\{
    -c4^kM^2n\varepsilon_{I,J}^{\,2}
  \}.
  \label{eq:proof-global-type-two}
\end{align}
Let \(\phi_n=\sup_{k\ge0}\phi_{n,k}\).  If \(M\) is sufficiently large,
the sum over \(k\) of the right side of
\eqref{eq:proof-global-type-one} tends to zero, so
\(P_0^{(n)}\phi_n\to0\).

On the event in \eqref{eq:proof-denominator-bound}, the posterior
probability in \eqref{eq:proof-predictor-contraction} is at most
\begin{align*}
  \phi_n
  &+
  e^{C_Dn\widetilde\varepsilon_{I,J}^{\,2}}
  \int_{
    \{d_n>M\varepsilon_{I,J}\}
    \cap\mathcal P_n
  }
  (1-\phi_n)
  \frac{dP_{\boldsymbol\vartheta}^{(n)}}
       {dP_0^{(n)}}(\mathbf Y)
  \,d\Pi(\boldsymbol\vartheta)\\
  &+
  e^{C_Dn\widetilde\varepsilon_{I,J}^{\,2}}
  \int_{\mathcal P_n^c}
  \frac{dP_{\boldsymbol\vartheta}^{(n)}}
       {dP_0^{(n)}}(\mathbf Y)
  \,d\Pi(\boldsymbol\vartheta).
\end{align*}
Taking \(P_0^{(n)}\)-expectations, using Fubini's theorem and
\eqref{eq:proof-global-type-two}, bounds the middle term by
\[
  e^{C_Dn\widetilde\varepsilon_{I,J}^{\,2}}
  \sum_{k\ge0}
  e^{-c4^kM^2n\varepsilon_{I,J}^{\,2}}
  \Pi(S_{n,k})
  =
  o(1),
\]
because
\(\widetilde\varepsilon_{I,J}\le\varepsilon_{I,J}\) and \(M\) is large.
The expectation of the last term is
\[
  e^{C_Dn\widetilde\varepsilon_{I,J}^{\,2}}
  \Pi(\mathcal P_n^c)
  =
  o(1)
\]
by \eqref{eq:proof-joint-tail}, after choosing \(C_0>C_D\).
Lemma~\ref{lem:proof-denominator}, Markov's inequality, and
\(P_0^{(n)}\phi_n\to0\) complete the proof.
\end{proof}

\paragraph{Transfer from the predictor to the difficulty surface.}

\begin{lemma}[Identification transfer]
\label{lem:proof-identification-transfer}
For every candidate and true parameter satisfying the exact centering
constraints in (A3),
\begin{equation}
  \|f-f_0\|_J
  \le
  d_n(\boldsymbol\eta,\boldsymbol\eta^0).
  \label{eq:proof-identification-transfer}
\end{equation}
\end{lemma}

\begin{proof}
For every item \(j\),
\begin{align*}
  \frac1I\sum_{i=1}^{I}
  (\eta_{ij}-\eta_{ij}^0)
  &=
  \alpha_j
  \left(\frac1I\sum_{i=1}^{I}\theta_i\right)
  -
  \alpha_j^0
  \left(\frac1I\sum_{i=1}^{I}\theta_i^0\right)\\
  &\quad
  -
  \{f(\mathbf s_j)-f_0(\mathbf s_j)\}\\
  &=
  -\{f(\mathbf s_j)-f_0(\mathbf s_j)\}.
\end{align*}
Jensen's inequality gives
\[
  \{f(\mathbf s_j)-f_0(\mathbf s_j)\}^2
  \le
  \frac1I\sum_{i=1}^{I}
  (\eta_{ij}-\eta_{ij}^0)^2.
\]
Averaging this inequality over \(j\) proves
\eqref{eq:proof-identification-transfer}.
\end{proof}

\begin{proof}[Proof of Theorem~\ref{thm:sgp_irt_contraction}]
Choose the logarithmic exponent in
$\varepsilon_{I,J}^{\mathrm{GP}}$ so that
\(\kappa\ge\kappa_{+}\) in $\kappa_{+}>\kappa_{-}+(D+1)/2$.  Then
\(\varepsilon_{I,J}^{\mathrm{item}}\) is at least a constant multiple of
\(\varepsilon_{I,J}\) in
\eqref{eq:proof-testing-rate}.  By
Lemmas~\ref{lem:proof-predictor-contraction} and
\ref{lem:proof-identification-transfer},
\[
  \Pi\!\left(
    \|f-f_0\|_J
    >
    M\varepsilon_{I,J}^{\mathrm{item}}
    \,\middle|\,
    \mathbf Y
  \right)
  \le
  \Pi\!\left(
    d_n(\boldsymbol\eta,\boldsymbol\eta^0)
    >
    M\varepsilon_{I,J}
    \,\middle|\,
    \mathbf Y
  \right)
  \longrightarrow0.
\]
\end{proof}

\subsection*{From the empirical item norm to continuum
\texorpdfstring{\(L^2\)}{L2}}
\label{app:continuum-transfer}

The empirical contraction theorem itself does not require quasi-uniform
item locations; quasi-uniformity is used only to infer recovery away from
the observed item sites.  The \(J^{-\nu/D}\) continuum term stated in the
main text requires the following explicit stable-reconstruction condition:
there is a constant \(C_{\mathrm{rec}}\) such that, for all posterior-sieve
differences \(g=f-f_0\),
\begin{equation}
  \|g\|_{L^2([0,1]^D)}
  \le
  C_{\mathrm{rec}}
  \left\{
    \|g\|_J+J^{-\nu/D}
  \right\}.
  \label{eq:proof-reconstruction}
\end{equation}
Under \eqref{eq:proof-reconstruction}, Theorem
\ref{thm:sgp_irt_contraction} immediately gives
\[
  \Pi\!\left(
    \|f-f_0\|_{L^2([0,1]^D)}
    >
    M
    \max\left\{
      \varepsilon_{I,J}^{\mathrm{item}},
      J^{-\nu/D}
    \right\}
    \,\middle|\,
    \mathbf Y
  \right)
  \longrightarrow0.
\]

For completeness, the next proposition records what follows from
quasi-uniformity and an ordinary H\"older bound without assuming a
higher-order reconstruction operator.

\begin{proposition}[Conservative quasi-uniform norm transfer]
\label{prop:proof-conservative-transfer}
Suppose \(g\in\mathcal C^\nu([0,1]^D;B_g)\) and the item design is
quasi-uniform.  Let \(\bar\nu=\nu\wedge1\).  Then
\begin{equation}
  \|g\|_{L^2([0,1]^D)}
  \le
  C\left\{
    \|g\|_J+B_gJ^{-\bar\nu/D}
  \right\}.
  \label{eq:proof-conservative-transfer}
\end{equation}
\end{proposition}

\begin{proof}
Let \(V_j\) be the Voronoi cell of \(\mathbf s_j\).  Quasi-uniformity gives
\(|V_j|\le C/J\), and every
\(\mathbf s\in V_j\) satisfies
\(\|\mathbf s-\mathbf s_j\|\le Ch_J\).  H\"older continuity yields
\[
  |g(\mathbf s)|
  \le
  |g(\mathbf s_j)|
  +
  CB_gh_J^{\bar\nu}.
\]
Therefore,
\begin{align*}
  \|g\|_{L^2}^2
  &=
  \sum_{j=1}^{J}
  \int_{V_j}|g(\mathbf s)|^2\,d\mathbf s\\
  &\le
  2\sum_{j=1}^{J}
  |V_j||g(\mathbf s_j)|^2
  +
  CB_g^2h_J^{2\bar\nu}
  \sum_{j=1}^{J}|V_j|\\
  &\le
  C\|g\|_J^2
  +
  CB_g^2J^{-2\bar\nu/D}.
\end{align*}
Taking square roots proves
\eqref{eq:proof-conservative-transfer}.
\end{proof}

\begin{remark}
For \(0<\nu\le1\), Proposition
\ref{prop:proof-conservative-transfer} already yields the main-text
discretization order \(J^{-\nu/D}\).  For \(\nu>1\), quasi-uniformity and a
nearest-site argument alone yield only \(J^{-1/D}\).  The faster
\(J^{-\nu/D}\) order then requires the stable order-\(\nu\) reconstruction
condition \eqref{eq:proof-reconstruction}, exactly as stated in the revised
theory section.
\end{remark}

\subsection{Proof of the fixed-category polytomous extension}
\label{app:polytomous-proof}

Use category \(0\) as the reference.  For
\(\boldsymbol\eta=(\eta_1,\ldots,\eta_{C-1})^\top\), write
\[
  p_c(\boldsymbol\eta)
  =
  \frac{e^{\eta_c}}
       {1+\sum_{k=1}^{C-1}e^{\eta_k}},
  \quad c=1,\ldots,C-1,
  \qquad
  p_0(\boldsymbol\eta)
  =
  \frac{1}
       {1+\sum_{k=1}^{C-1}e^{\eta_k}}.
\]
Define
\begin{equation}
  d_{n,C}^2(
    \boldsymbol\eta,\boldsymbol\eta^0
  )
  =
  \frac1{IJ}
  \sum_{i=1}^{I}\sum_{j=1}^{J}
  \sum_{c=1}^{C-1}
  (\eta_{ijc}-\eta_{ijc}^0)^2.
  \label{eq:proof-polytomous-metric}
\end{equation}

\begin{lemma}[Uniform reference-softmax geometry]
\label{lem:proof-softmax-geometry}
Suppose \(C\) is fixed and
\(\boldsymbol\eta,\boldsymbol\eta^0\in[-L,L]^{C-1}\). There exist
constants \(0<c_{L,C}<C_{L,C}<\infty\) such that the one-observation
KL divergence and centered log-likelihood variance are bounded above by
\(
  C_{L,C}
  \|\boldsymbol\eta-\boldsymbol\eta^0\|_2^2,
\)
and the squared Hellinger distance satisfies
\[
  c_{L,C}
  \|\boldsymbol\eta-\boldsymbol\eta^0\|_2^2
  \le
  h^2(\boldsymbol\eta^0,\boldsymbol\eta)
  \le
  C_{L,C}
  \|\boldsymbol\eta-\boldsymbol\eta^0\|_2^2.
\]
\end{lemma}

\begin{proof}
For \(\boldsymbol\eta\in\mathbb R^{C-1}\), define
\[
  A_C(\boldsymbol\eta)
  =
  \log\left(
    1+\sum_{c=1}^{C-1}e^{\eta_c}
  \right),
\]
and let
\(
  p_0(\boldsymbol\eta)
  =
  e^{-A_C(\boldsymbol\eta)},
  \
  p_c(\boldsymbol\eta)
  =
  e^{\eta_c-A_C(\boldsymbol\eta)},
  \ c=1,\ldots,C-1.
\)
Thus category \(0\) is the reference category. If
\(\boldsymbol\eta\in[-L,L]^{C-1}\), then
\[
  1+\sum_{c=1}^{C-1}e^{\eta_c}
  \le
  1+(C-1)e^L.
\]
Consequently, all \(C\) probabilities are uniformly bounded away from
zero:
\[
  p_c(\boldsymbol\eta)\ge m_{L,C},
  \qquad c=0,\ldots,C-1,
\]
where, for example, one may take
\[
  m_{L,C}
  =
  \frac{e^{-L}}{1+(C-1)e^L}>0.
\]

Write
\(
  \mathbf p(\boldsymbol\eta)
  =
  \bigl(
    p_1(\boldsymbol\eta),\ldots,p_{C-1}(\boldsymbol\eta)
  \bigr)^\top.
\)
Direct differentiation gives
\(
  \nabla A_C(\boldsymbol\eta)
  =
  \mathbf p(\boldsymbol\eta)
\)
and
\(
  H(\boldsymbol\eta)
  :=
  \nabla^2A_C(\boldsymbol\eta)
  =
  \operatorname{diag}\{\mathbf p(\boldsymbol\eta)\}
  -
  \mathbf p(\boldsymbol\eta)
  \mathbf p(\boldsymbol\eta)^\top.
\)

We first establish uniform lower and upper bounds for this Hessian. Let
\(\mathbf v=(v_1,\ldots,v_{C-1})^\top\in\mathbb R^{C-1}\), append
\(v_0=0\), and define
\[
  \overline v_p
  =
  \sum_{c=0}^{C-1}p_c(\boldsymbol\eta)v_c
  =
  \sum_{c=1}^{C-1}p_c(\boldsymbol\eta)v_c.
\]
Then
\begin{align*}
  \mathbf v^\top H(\boldsymbol\eta)\mathbf v
  &=
  \sum_{c=1}^{C-1}p_c(\boldsymbol\eta)v_c^2
  -
  \left\{
    \sum_{c=1}^{C-1}p_c(\boldsymbol\eta)v_c
  \right\}^2                                                     \\
  &=
  \sum_{c=0}^{C-1}
  p_c(\boldsymbol\eta)
  \bigl(v_c-\overline v_p\bigr)^2                                 \\
  &=
  \frac12
  \sum_{c=0}^{C-1}\sum_{k=0}^{C-1}
  p_c(\boldsymbol\eta)p_k(\boldsymbol\eta)(v_c-v_k)^2.
\end{align*}
Retaining only the terms involving the reference category \(0\) gives
\begin{align*}
  \mathbf v^\top H(\boldsymbol\eta)\mathbf v
  &\ge
  p_0(\boldsymbol\eta)
  \sum_{c=1}^{C-1}p_c(\boldsymbol\eta)v_c^2                         \\
  &\ge
  m_{L,C}^2
  \sum_{c=1}^{C-1}v_c^2
  =
  m_{L,C}^2\|\mathbf v\|_2^2.
\end{align*}
For the upper bound, the nonnegative squared-mean term may be discarded:
\begin{align*}
  \mathbf v^\top H(\boldsymbol\eta)\mathbf v
  &=
  \sum_{c=1}^{C-1}p_c(\boldsymbol\eta)v_c^2
  -
  \left\{
    \sum_{c=1}^{C-1}p_c(\boldsymbol\eta)v_c
  \right\}^2                                                       \\
  &\le
  \sum_{c=1}^{C-1}p_c(\boldsymbol\eta)v_c^2                         \\
  &\le
  \sum_{c=1}^{C-1}v_c^2
  =
  \|\mathbf v\|_2^2.
\end{align*}
Therefore, uniformly over \(\boldsymbol\eta\in[-L,L]^{C-1}\),
\[
  m_{L,C}^2 I_{C-1}
  \preceq
  H(\boldsymbol\eta)
  \preceq
  I_{C-1}.
\]
In particular,
\(
  \sup_{\boldsymbol\eta\in[-L,L]^{C-1}}
  \|H(\boldsymbol\eta)\|_{\mathrm{op}}
  \le 1.
\)

Let
\(
  \boldsymbol\delta
  =
  \boldsymbol\eta-\boldsymbol\eta^0,
  \
  \boldsymbol\eta_t
  =
  \boldsymbol\eta^0+t\boldsymbol\delta,
  \  0\le t\le1.
\)
Because the cube is convex,
\(\boldsymbol\eta_t\in[-L,L]^{C-1}\) for every \(t\in[0,1]\). The
one-observation KL divergence from
\(P_{\boldsymbol\eta^0}\) to \(P_{\boldsymbol\eta}\) is the Bregman
divergence generated by \(A_C\):
\[
  K(
    \boldsymbol\eta^0,\boldsymbol\eta
  )
  =
  A_C(\boldsymbol\eta)
  -
  A_C(\boldsymbol\eta^0)
  -
  \nabla A_C(\boldsymbol\eta^0)^\top
  \boldsymbol\delta.
\]
Taylor's theorem with integral remainder yields
\[
  K(
    \boldsymbol\eta^0,\boldsymbol\eta
  )
  =
  \int_0^1
  (1-t)
  \boldsymbol\delta^\top
  H(\boldsymbol\eta_t)
  \boldsymbol\delta
  \,dt.
\]
The uniform Hessian bounds therefore imply the stronger two-sided
inequality
\[
  \frac{m_{L,C}^2}{2}
  \|\boldsymbol\delta\|_2^2
  \le
  K(
    \boldsymbol\eta^0,\boldsymbol\eta
  )
  \le
  \frac12
  \|\boldsymbol\delta\|_2^2.
\]

We next consider the centered log-likelihood variance. Extend
\(\boldsymbol\delta\) by setting \(\delta_0=0\). For an observation
\(Y\in\{0,\ldots,C-1\}\),
\[
  \log
  \frac{p_Y(\boldsymbol\eta^0)}
       {p_Y(\boldsymbol\eta)}
  =
  -\delta_Y
  +
  A_C(\boldsymbol\eta)-A_C(\boldsymbol\eta^0).
\]
The second term is nonrandom. Hence, under
\(Y\sim P_{\boldsymbol\eta^0}\),
\begin{align*}
  V(
    \boldsymbol\eta^0,\boldsymbol\eta
  )
  &:=
  \operatorname{Var}_{\boldsymbol\eta^0}
  \left[
    \log
    \frac{p_Y(\boldsymbol\eta^0)}
         {p_Y(\boldsymbol\eta)}
  \right]                                                         \\
  &=
  \operatorname{Var}_{\boldsymbol\eta^0}(\delta_Y)                 \\
  &=
  \boldsymbol\delta^\top
  H(\boldsymbol\eta^0)
  \boldsymbol\delta                                                \\
  &\le
  \|\boldsymbol\delta\|_2^2.
\end{align*}

It remains to establish the Hellinger bounds. Define the square-root
probability map
\(
  \boldsymbol\rho(\boldsymbol\eta)
  =
  \bigl(
    \sqrt{p_0(\boldsymbol\eta)},\ldots,
    \sqrt{p_{C-1}(\boldsymbol\eta)}
  \bigr)^\top.
\)
Using the convention
\(
  h^2(
    \boldsymbol\eta^0,\boldsymbol\eta
  )
  =
  \|
    \boldsymbol\rho(\boldsymbol\eta)
    -
    \boldsymbol\rho(\boldsymbol\eta^0)
  \|_2^2,
\)
we first prove the upper bound. For
\(\mathbf v\in\mathbb R^{C-1}\), again set \(v_0=0\). The directional
derivative of the \(k\)th component of \(\boldsymbol\rho\) is
\[
  D\rho_k(\boldsymbol\eta)[\mathbf v]
  =
  \frac12
  \sqrt{p_k(\boldsymbol\eta)}
  \left(
    v_k-\sum_{c=1}^{C-1}p_c(\boldsymbol\eta)v_c
  \right),
  \quad k=0,\ldots,C-1.
\]
It follows that
\begin{align*}
  \|
    D\boldsymbol\rho(\boldsymbol\eta)[\mathbf v]
  \|_2^2
  &=
  \frac14
  \sum_{k=0}^{C-1}
  p_k(\boldsymbol\eta)
  \left(
    v_k-\sum_{c=1}^{C-1}p_c(\boldsymbol\eta)v_c
  \right)^2                                                       \\
  &=
  \frac14
  \mathbf v^\top
  H(\boldsymbol\eta)
  \mathbf v                                                       \\
  &\le
  \frac14\|\mathbf v\|_2^2.
\end{align*}
Thus
\(
  \|D\boldsymbol\rho(\boldsymbol\eta)\|_{\mathrm{op}}
  \le 1/2
\)
uniformly in \(\boldsymbol\eta\). Applying the fundamental theorem of
calculus along the segment \(\boldsymbol\eta_t\),
\begin{align*}
  \|
    \boldsymbol\rho(\boldsymbol\eta)
    -
    \boldsymbol\rho(\boldsymbol\eta^0)
  \|_2
  &=
  \left\|
    \int_0^1
    D\boldsymbol\rho(\boldsymbol\eta_t)
    \boldsymbol\delta\,dt
  \right\|_2                                                     \\
  &\le
  \int_0^1
  \|D\boldsymbol\rho(\boldsymbol\eta_t)\|_{\mathrm{op}}
  \|\boldsymbol\delta\|_2\,dt                                    \\
  &\le
  \frac12\|\boldsymbol\delta\|_2.
\end{align*}
Therefore,
\(
  h^2(
    \boldsymbol\eta^0,\boldsymbol\eta
  )
  \le
  4^{-1}\|\boldsymbol\delta\|_2^2.
\)

For the reverse inequality, observe that
\[
  \eta_c
  =
  \log\frac{p_c(\boldsymbol\eta)}
                {p_0(\boldsymbol\eta)}
  =
  2\left\{
    \log\rho_c(\boldsymbol\eta)
    -
    \log\rho_0(\boldsymbol\eta)
  \right\}.
\]
Because every probability is at least \(m_{L,C}\), every coordinate of
\(\boldsymbol\rho(\boldsymbol\eta)\) and
\(\boldsymbol\rho(\boldsymbol\eta^0)\) is at least
\(\sqrt{m_{L,C}}\). The mean value theorem for the logarithm therefore
gives
\[
  |\log x-\log y|
  \le
  \frac{|x-y|}{\sqrt{m_{L,C}}}
  \qquad
  \text{for }x,y\in[\sqrt{m_{L,C}},1].
\]
Hence, for each \(c=1,\ldots,C-1\),
\begin{align*}
  |\delta_c|
  &\le
  \frac{2}{\sqrt{m_{L,C}}}
  \left(
    |\rho_c(\boldsymbol\eta)-\rho_c(\boldsymbol\eta^0)|
    +
    |\rho_0(\boldsymbol\eta)-\rho_0(\boldsymbol\eta^0)|
  \right),
\end{align*}
and therefore
\[
  \delta_c^2
  \le
  \frac{8}{m_{L,C}}
  \left[
    \{\rho_c(\boldsymbol\eta)-\rho_c(\boldsymbol\eta^0)\}^2
    +
    \{\rho_0(\boldsymbol\eta)-\rho_0(\boldsymbol\eta^0)\}^2
  \right].
\]
Summing over \(c=1,\ldots,C-1\) gives
\begin{align*}
  \|\boldsymbol\delta\|_2^2
  &\le
  \frac{8}{m_{L,C}}
  \left[
    \sum_{c=1}^{C-1}
    \{\rho_c(\boldsymbol\eta)-\rho_c(\boldsymbol\eta^0)\}^2
    +(C-1)
    \{\rho_0(\boldsymbol\eta)-\rho_0(\boldsymbol\eta^0)\}^2
  \right]                                                        \\
  &\le
  \frac{8(C-1)}{m_{L,C}}
  \|
    \boldsymbol\rho(\boldsymbol\eta)
    -
    \boldsymbol\rho(\boldsymbol\eta^0)
  \|_2^2.
\end{align*}
Thus,
\(
  h^2(
    \boldsymbol\eta^0,\boldsymbol\eta
  )
  \ge
  ({m_{L,C}}/{8(C-1)})
  \|\boldsymbol\delta\|_2^2.
\)
Combining the preceding bounds proves the result. If squared Hellinger
distance is defined with an additional factor \(1/2\), the same
conclusion holds after rescaling the constants.
\end{proof}

\begin{proof}[Proof of Corollary~\ref{cor:poly_contraction}]
Lemma \ref{lem:proof-softmax-geometry} supplies the KL, variance, Hellinger,
and testing bounds in the metric \(d_{n,C}\).  Because \(C\) is fixed,
independent hierarchical GP priors on
\((f_1,\ldots,f_{C-1})\) multiply the prior concentration exponent and
sieve entropy by only the fixed factor \(C-1\).  The nuisance entropy is
unchanged up to a constant depending on
\(\max_{c<C}|q_c|\).  Repeating
Lemmas~\ref{lem:proof-joint-prior-mass}--\ref{lem:proof-predictor-contraction}
therefore gives
\[
  \Pi\!\left\{
    d_{n,C}(
      \boldsymbol\eta,\boldsymbol\eta^0
    )
    >
    M\varepsilon_{I,J}^{\mathrm{item}}
    \,\middle|\,
    \mathbf Y
  \right\}
  \longrightarrow0.
\]

For every \(j\) and every nonreference category \(c\),
\begin{align*}
  \frac1I\sum_{i=1}^{I}
  (\eta_{ijc}-\eta_{ijc}^0)
  &=
  q_c\alpha_j
  \left(\frac1I\sum_{i=1}^{I}\theta_i\right)
  -
  q_c\alpha_j^0
  \left(\frac1I\sum_{i=1}^{I}\theta_i^0\right)\\
  &\quad
  -
  \{f_c(\mathbf s_j)-f_{0c}(\mathbf s_j)\}\\
  &=
  -\{f_c(\mathbf s_j)-f_{0c}(\mathbf s_j)\}.
\end{align*}
Jensen's inequality, followed by summation over \(j\) and \(c\), gives
\[
  \sum_{c=1}^{C-1}
  \|f_c-f_{0c}\|_J^2
  \le
  d_{n,C}^2(
    \boldsymbol\eta,\boldsymbol\eta^0
  ).
\]
\end{proof}